\documentclass[smallextended]{svjour3}
\usepackage[utf8]{inputenc}
\usepackage{color,graphicx}
\usepackage{amsmath,amssymb}
\usepackage{thmtools,thm-restate}
\usepackage{natbib}
\setcitestyle{authoryear}

\usepackage[margin=1in]{geometry}

\providecommand{\dist}{\mathcal{D}}

\providecommand{\shapley}{\varphi}
\providecommand{\shapleyx}{\shapley^{[\mathrm{x}]}}
\providecommand{\mixture}[2]{#1\lessgtr#2}
\providecommand{\intervals}{\mathcal{I}}
\DeclareMathOperator{\Res}{Res}
\DeclareMathOperator*{\Ex}{\mathbb{E}}
\DeclareMathOperator*{\var}{\mathbb{V}}
\DeclareMathOperator{\Exp}{Exp}

\providecommand{\jo}[1]{}
\providecommand{\yf}[1]{}
\providecommand{\ks}[1]{}
\providecommand{\NB}[1]{}
\providecommand{\qedhere}{}

\begin{document}
\title{Shapley Values in Weighted Voting Games with Random Weights}

\author{Yuval Filmus \and Joel Oren \and Kannan Soundararajan}
\institute{
Yuval Filmus \at Technion --- Israel Institute of Technology, Haifa, Israel, \\\email{yuvalfi@cs.technion.ac.il} \and
Joel Oren \at University of Toronto, Toronto, ON, Canada, \\\email{oren@cs.toronto.edu} \and
Kannan Soundararajan \at Stanford University, Stanford, CA, \\\email{ksound@stanford.edu}
}

\maketitle

\begin{abstract}
We investigate the distribution of the well-studied Shapley--Shubik values in weighted voting games
where the agents are stochastically determined.
The Shapley--Shubik value measures the voting power of an  agent, in
typical collective decision making
systems. While easy to estimate empirically given the parameters of a
weighted voting game, the Shapley values are notoriously hard to reason about analytically.

We propose a probabilistic approach in
which the agent weights are drawn i.i.d.\ from some known exponentially decaying
distribution. We provide a general closed-form characterization of the highest
and lowest expected Shapley values in such a game, as a function of the
parameters of the underlying distribution. To do so, we give a novel reinterpretation of
the stochastic process that generates the Shapley variables as a
renewal process. We demonstrate the use of our results on the uniform
and exponential distributions. Furthermore, we show the strength of our
theoretical predictions on several synthetic datasets.
\end{abstract}

 %\MSC[2010] 91A12

\begin{acknowledgements}
We thank Noga Alon for helpful
discussions. We would also like to thank Yoram Bachrach and Yair Zick for
their comments and suggestions.

YF did part of the work while at the University of Toronto and the Institute for Advanced Study, Princeton, NJ. This material is based upon work supported by the National Science Foundation under agreement No.~DMS-1128155. Any opinions, findings and conclusions or recommendations expressed in this material are those of the authors, and do not necessarily reflect the views of the National Science Foundation.

KS is supported in part by grants from the National Science Foundation, and 
a Simons Investigator grant from the Simons Foundation.
\end{acknowledgements}

\newpage
 
corollary\section{Introduction} \label{sec:intro}
Weighted voting games are a fundamental model in cooperative game
theory, used to formalize collective decision making processes.
A weighted voting game is given by a set of $n$
positively-weighted agents, and a positive quota value, also known as the threshold. A set of
agents (a coalition) is said to be  \emph{winning} if their combined weight exceeds the
prescribed quota. This setting was inspired by parliamentary
systems. Indeed, when forming a coalition
following an election, or when legislating a law, there is a minimal
threshold (e.g., the minimal number of parliament seats needed to form
a coalition, or the necessary amount of votes required to pass a
bill), which requires the agents to
cooperate. Due to its direct relevance to such political bodies, much
work has been done on modelling real-life parliamentary systems as
weighted voting games, and studying their properties from this
perspective. Some examples include the European
Council of Ministers \citep[e.g.,][]{leech02vote,algaba:2007,machover:2004},
the International Monetary Fund studied by Leech~\citep{leech2002voting}, and
the United Nations Security Council~\citep{strand2011weighted}.
Furthermore, weighted voting games have also been used to formalize
other decision-making processes in companies, where the votes are cast by
the shareholders~\citep[e.g.,][]{arcaini86}.

A recurring thread in many of these works is the study of the voting power of agents in
these systems. To that end, a number of measures called ``power
indices'' have been proposed. Informally, a power index quantifies the amount of influence
each agent has in a particular game.

Two well-known examples of such power indices are
the Shapley power index~\citep{marketgames}\footnote{The more specific
  term for our case is the Shapley--Shubik, which pertains to the
  special case of the Shapley value in weighted voting games. We use
  the term Shapley value for succinctness purposes only.} and the Banzhaf
index~\citep{banzhaf}.
Both of these power indices build on the notion of
 a pivotal agent: an agent is pivotal to a given coalition $A$, if $A$
 is a losing coalition (the combined weight of its agents does not
 exceed the quota), and upon joining it, the coalition becomes
 a winning one. A particular use of these measures pertains to the design of voting
systems, where the goal is usually to reach a certain state of power distribution, by
either assigning the agents their weights, or by judiciously selecting
the quota. The goal of this paper is very much related to the latter approach. That
is, we are interested in analyzing the effects of varying the quota
on the resulting Shapley values.

In general though, it is well-known that weighted voting games are
hard to reason about. The task of computing the Shapley
values exactly is known to be computationally intractable
\citep{Matsui2001305}. More importantly, for an arbitrary weighted
voting game, the values can exhibit acute fluctuations with
even minute changes to the quota value.
Given this fact, we revisit the model by considering the manner by
which the weights are determined. We assume that the weights are
the result of an underlying stochastic process. This
additional layer is relevant to many of the settings that
our theoretical model applies to. For example, in the case of
a parliamentary system, the number of seats each party depends greatly
on the size of the population that it represents.
However, it is often assumed that the sizes of populations are not
fixed, but rather,  vary stochastically
 for various reasons~\citep[e.g.,][]{matis2000stochastic}.
In accordance with this added distributional assumption, we revisit
the study of Shapley values by considering the \emph{expected} Shapley
values. More precisely,we ask the following question:
\emph{Assuming that weights are drawn i.i.d.\ from a certain distribution, can we give a
closed-form description of the expected Shapley values, as a function of the quota?
}
In this paper we assume that the agent weights are drawn i.i.d.\
from an \emph{exponentially decaying} distribution. This assumption allows us to
provide a strong characterization of the highest and lowest
Shapley values. We accomplish this using a novel connection to renewal theory.
To illustrate our general result, we apply it to a number of well-known
distributions. We empirically demonstrate our findings using simulations.

\paragraph{Structure of the paper} We outline some of the relevant
previous work in Section~\ref{sec:prev_work}. In Section~\ref{sec:defs} we
provide the necessary notation, definitions, and some basic properties of the
Shapley value. In Section~\ref{sec:main}, we formally present our main
problem of interest, along with an empirical illustration, based on simulation
results for the case of the uniform distribution on the unit
interval. We also present our main result (Theorem~\ref{thm:iid}) in
that section, that pertains to the two extreme agents (those with the highest and
lowest weights).

Before we prove Theorem~\ref{thm:iid}, we first apply it
to a number of interesting distributions. We begin with the case of
the uniform distribution in Section~\ref{sec:uniform}. This provides a
theoretic validation of our empirical results given in
Section~\ref{sec:main}. We tackle the case of the exponential
distribution in Section~\ref{sec:exponential}.
We then prove Theorem~\ref{thm:iid} in Section~\ref{sec:main-proof}, and discuss conjectural extensions in Section~\ref{sec:other}. Our proof uses a technical result on renewal processes that may be of independent interest (Proposition~\ref{pro:renewal}), which is proved in Section~\ref{sec:sound}. Finally, we provide concluding remarks and directions for future research in Section~\ref{sec:conclusions}.

% \paragraph{Acknowledgements} We thank Noga Alon for helpful
% discussions. We would also like to thank Yoram Bachrach and Yair Zick for
% their comments and suggestions.
% 
% KS is supported in part by grants from the National Science Foundation, and 
% a Simons Investigator grant from the Simons Foundation.
% 
% This material is based upon work supported by the National Science Foundation under agreement No.~DMS-1128155. Any opinions, findings and conclusions or recommendations expressed in this material are those of the authors, and do not necessarily reflect the views of the National Science Foundation.

%%% Local Variables:
%%% mode: latex
%%% TeX-master: "wvg-iid"
%%% End:

\section{Related Work} \label{sec:prev_work}
Ample work has been done on characterizing the power distribution in weighted voting games.
The focus of this paper is the study of the Shapley--Shubik value \citep{marketgames}, which is a special case
of the more general Shapley value \citep{shapleyvalue}. For succinctness, we use the term Shapley value for the special case of weighted voting games as well. Another well-studied measure is the Banzhaf power index \citep{banzhaf}. See \citet{felsen98book} for a survey of these power indices.

Shapley values pose a number of difficulties: from a computational perspective, Shapley values have been shown to be hard to compute; see \citet{ocfgeb} for details. That said, they can easily be approximated using sampling techniques~\citep[e.g.,][]{mann62sv,bachrach10approx}. %\footnote{Note that the approximation relies on the probabilistic interpretation of the above two power indices, and assumes arbitrary agent weights.}
Moreover, power indices have been shown to exhibit considerable volatility to changes in the quota, in the absence of any structural guarantees on the weights \citep[e.g.,][]{zuck12manip,zick13var}. \citet{zick2011sv} demonstrats the effect of changes in the quota to the agents' Shapley values under a uniform distribution of agent weights. Recent work by \citet{oren2014power} studies similar questions under different probabilistic models of the weights, and in the case of (fixed) super-increasing weight vectors.

\citet{jelnov2012sv} consider the case in which the agent weights are drawn i.i.d.\ from an exponential distribution. They show that the Shapley value of an agent is proportional to its weight \emph{in expectation}, for a wide range of quota values. Their method can potentially be used to analyze other distributions. The main difference between their work and ours is that they consider the \emph{unsorted} Shapley values while we consider the \emph{sorted} Shapley values. It seems that for technical reasons, their line of reasoning cannot be used to derive our results.

Finally, there has been a growing body of literature on the so-called inverse Shapley value problem of selecting the weights vector so as to obtain a resulting vector of desired power indices \citep[e.g.,][]{de2010enumeration,servedio12inversewvg,aziz2007efficient}. This line of work can be thought of as being tangential to our study.

%%% Local Variables:
%%% mode: latex
%%% TeX-master: "wvg-iid"
%%% End:

\section{Definitions} \label{sec:defs}
\paragraph{General notation} We use the notation $[n] = \{1,\ldots,n\}$. A permutation $\pi$ is a one-to-one mapping $\pi\colon [n] \rightarrow [n]$. The symmetric group on $[n]$ (the set of all permutations of $[n]$) is denoted $S_n$. The notation $U(a,b)$ signifies the uniform distribution over the interval $[a,b]$. The notation $\Exp(\lambda)$ signifies the exponential distribution with density $\lambda e^{-\lambda t}$ and mean $1/\lambda$. %The notation $\Bin(n,p)$ signifies the binomial distribution with $n$ trials and success probability $p$.

For a random variable $Z$ and a real number $t$, $Z_{\leq t}$ is the distribution of $Z$ conditioned on being at most $t$, and $Z_{\geq t}$ is the distribution of $Z$ conditioned on being at least $t$. Given an integer $n$, we let the random variable $Z^n_{\max}$ denote $\max(Z_1,\ldots,Z_n)$, where $Z_1,\ldots,Z_n$ are $n$ independent copies of $Z$; we define $Z^n_{\min}$ analogously.

\paragraph{Weighted voting games} A \emph{weighted voting game} (WVG) is specified by a set of agents $N=[n]$, with a corresponding non-decreasing sequence of assigned weights $w_1,\ldots,w_n$, and a quota $q \geq 0$. We think of the weights as being fixed and of the quota as a parameter. A subset (a \emph{coalition}) $S \subseteq [n]$ is said to be \emph{winning} if $w(S)=\sum_{i \in S} w_i \geq q$; otherwise, it is a losing coalition. An agent $i \in N$ is said to be \emph{pivotal} with respect to a coalition $S \subseteq N$ if $S \cup \{i\}$ is winning, whereas $S$ is not. In other words, $i$ is pivotal if $q \in (w(S),w(S \cup \{i\})]$. The Shapley value~\citep{shapleyshubik} of an agent $i \in N$ is the probability that, if one selects a random permutation $\pi \in S_n$ uniformly at random, agent $i$ would be pivotal with respect to its predecessors in the permutation. Formally, the Shapley value of agent $i$ is defined as follows:
\begin{align*}
 \shapley_i(q) &= %\Pr_{\pi \in S_n}[\{ \pi(1),\ldots,\pi(\pi^{-1}(i)-1) \} \text{ is losing but } \{ \pi(1),\ldots,\pi(\pi^{-1}(i)) \} \text{ is winning}] \\ &=
 \Pr_{\pi \in S_n}[w_{\pi(1)} + \cdots + w_{\pi(\pi^{-1}(i)-1)} < q \leq w_{\pi(1)} + \cdots + w_{\pi(\pi^{-1}(i))}] \\ &=
 \Pr_{\pi \in S_n}[q - w_i \leq w_{\pi(1)} + \cdots + w_{\pi(\pi^{-1}(i)-1)} < q],
\end{align*}
where $\pi(j)$ denotes the agent at position $j$ in the permutation $\pi$.

Some elementary properties of the Shapley value include:
\begin{itemize}
 \item If $i \leq j$ (and so $w_i \leq w_j$) then $\shapley_i(q) \leq \shapley_j(q)$.
 \item For $0 \leq q \leq w_1+\cdots+w_n$ we have $\sum_i \shapley_i(q) = 1$.
\end{itemize}
%%% Local Variables:
%%% mode: latex
%%% TeX-master: "wvg-iid"
%%% End:

\section{Main result and applications} \label{sec:main}
Consider a weighted voting game in which the weights are sampled independently from some ``reasonable'' continuous distribution $\dist$. We call this the \emph{natural iid model}. In this paper we analyze the expected largest and smallest Shapley values in terms of the continuous distribution $\dist$, for all values of the quota bounded away from $0$ and $n\Ex[\dist]$.

Another natural model to consider is very similar: the weights are sampled independently from $\dist$, and then normalized to have unit sum. We call this the \emph{normalized iid model}. Simulation results show that both models behave very similarly. We analyze the natural iid model since it is easier to work with.

In order to estimate the Shapley value of the highest-weighted agent $n$, it is not hard to see that an answer to the following question would prove instrumental:

\medskip

\emph{Conditioning on them being at most the highest weight $w_n$, let $Y_1,\ldots,Y_{n-1} \sim \dist$, and define $S_m=\sum_{i=1}^mY_i$. What is the expected number of points from $\{S_1,\ldots,S_{n-1}\}$ that lie in the interval $[q-w_n,q)$?}

\medskip

\noindent A symmetric question can be phrased for the lowest-weighted agent as well. Put in these terms, we can think of the process that generates the $n-1$ weights (conditioning on the highest weight) as a renewal process. Roughly speaking, a renewal process \citep[see e.g.,][]{gallager1995discrete} is defined by sequential arrivals, where the gap between every two arrivals (the \emph{inter-arrival} times) are stochastic. In our case, the inter-arrival times are given by the agent weights, and the measure in question is the expected number of arrivals within the specified interval. Our analysis gives a powerful characterization of the Shapley values of both the highest- and lowest-weighted agents.

Before we begin to analyze the Shapley values, we consider the special case of the uniform distribution, for which our simulation results (for the normalized iid model) are depicted in Figure~\ref{fig:uniform}. Intuitively, we can see that apart from two relatively short intervals at the two extremes of the interval $[0,1]$, the Shapley values are stable at $2/n$ for the highest Shapley value, and roughly $2/n^2$ (this estimate will be justified momentarily) for the lowest Shapley value. From a more practical point of view, this means that as the number of agents increases, the ratio of the highest to lowest Shapley values grows at a \emph{linear} rate.
\begin{figure}%[H]
	\centering
	\includegraphics[width=.8\textwidth]{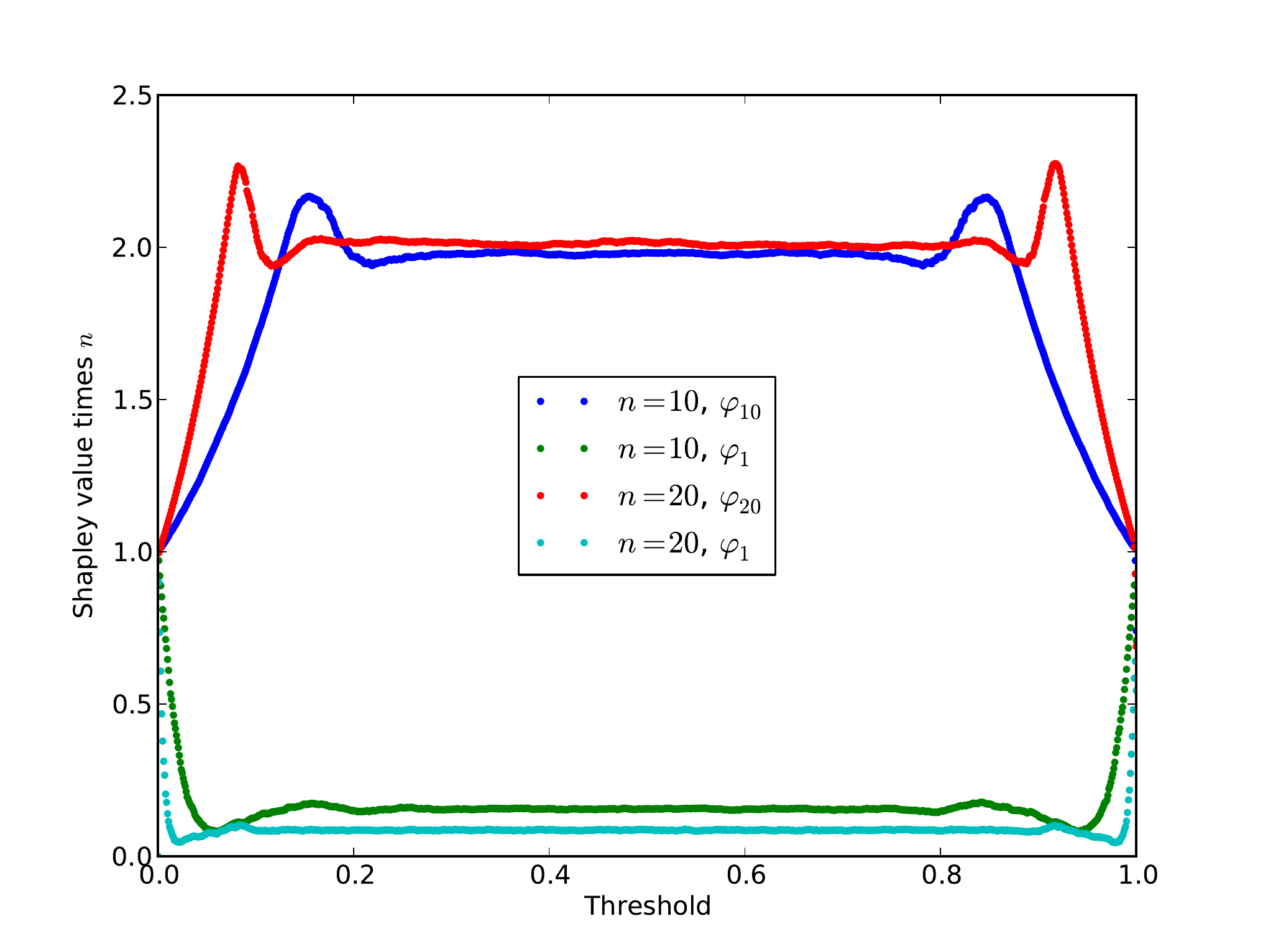} \caption{Shapley values for $X = U(0,1)$ and $n = 10,20$ of both minimal and maximal agents, multiplied by $n$, for the normalized iid model. Results of $10^6$ experiments.} \label{fig:uniform}
\end{figure}

Given the above results, we proceed to a rigorous analysis of the two extreme Shapley values.
Let $X$ be a non-negative continuous random variable whose density function $f$ has exponential decay: $f(t) \leq Ce^{-\lambda t}$ for some $\lambda > 0$. Note that any \emph{bounded} continuous random variable trivially has exponential decay, and that the exponential decay guarantees that $X$ has moments of all orders. Let $\chi_{\min}$ be the infimal $x$ such that $\Pr[X \geq x] > 0$, and let $\chi_{\max}$ be the supremal $x$ such that $\Pr[X \leq x] > 0$ (possibly $\chi_{\max} = \infty$). For example, the uniform distribution $U(a,b)$ has $\chi_{\min} = a$ and $\chi_{\max} = b$, and the exponential distribution $\Exp(\lambda)$ has $\chi_{\min} = 0$ and $\chi_{\max} = \infty$.

For technical reasons, we need to make an additional assumption on $X$: that the support of $X$ can be partitioned into finitely many intervals, in each of which the density $f$ is differentiable and is monotone (non-decreasing or non-increading). This assumption is satisfied by our main examples, the uniform and exponential distributions. We believe that this assumption can be weakened.

Another helpful technical assumption is that $x = O(\Ex[X_{\leq x}])$ for $x$ near $\chi_{\min}$, or more formally, for some $C,\epsilon > 0$ we have $x \leq C\Ex[X_{\leq x}]$ for all $x \in (\chi_{\min},\chi_{\min}+\epsilon)$.
%Although our results are meaningful even without this assumption, they become cleaner to state with it.
While it is possible to construct random variables invalidating this property, all natural random variables do satisfy the property. For example, the property trivially holds when $\chi_{\min} > 0$, and if $\chi_{\min} = 0$, then it holds when $f(x) = O(x^\delta)$ near zero for any $\delta > 0$, and in particular when $f$ is real analytic near zero. In our applications we verify this property explicitly.

Consider the following random process for generating weights (the \emph{natural iid model}): sample $x_1,\ldots,x_n$ from $X$ independently. We define the weights $w_1,\ldots,w_n$ to be the sequence obtained from $x_1,\ldots,x_n$ by sorting them in non-decreasing order, i.e., $w_1$ is the smallest weight and $w_n$ is the largest weight.

The following is our main theorem. We defer its proof to Section~\ref{sec:main-proof} in favor of discussing a couple of its implications.
%\yf{State and prove a simplified version of the first formula when $\chi_{\min} > 0$. Then use it to analyze general uniform distributions.}
%\jo{To resolve: $\varphi_n,\varphi_1$ are nowhere defined as order statistics.} \yf{The weights are sorted. I commented on that above.}
\begin{restatable}{theorem}{thmiid}
% lower bounds were \max(2m,4\lambda^{-1}\log n) and m
	\label{thm:iid}
Let $X$ be a non-negative continuous random variable whose density function $f$ satisfies $f(t) \leq Ce^{-\lambda t}$ for some $C,\lambda > 0$, and additionally, the support of $f$ can be partitioned into finitely many intervals on which $f$ is differentiable and monotone, and that $x = O(\Ex_{X \leq x}[X])$ for $x$ near $\chi_{\min}$.

For all $\epsilon > 0$ there exists $\xi < 1$ such that the following holds.
For all $Q \in [n^{1/4},(1-\epsilon)\Ex [X]]$,
	\[ \Ex[\shapleyx_{\max}(Q)] = \frac{1}{n} \Ex_{x \sim X^n_{\max}} \big[\frac{x}{\Ex [X_{\leq x}]}\big] \pm O(\xi^{n^{1/4}}). \]
Similarly, for all $Q \in [n^{1/4}, (n-n^{2/3})\Ex [X]]$,
	\[ \Ex[\shapleyx_{\min}(Q)] = \frac{1}{n} \Ex_{x \sim X^n_{\min}} \big[\frac{x}{\Ex [X_{\geq x}]}\big] \pm O(\xi^{n^{1/4}}). \]
\end{restatable}
% \begin{restatable}{theorem}{thmiid}
% %\begin{theorem}
% 	\label{thm:iid}
% Let $X$ be a non-negative continuous random variable whose density function $f$ satisfies $f(t) \leq Ce^{-\lambda t}$ for some $C,\lambda > 0$, and furthermore (for the first statement) $x = O(\Ex[X_{\leq x}])$ for $x$ near $\chi_{\min}$, where $\chi_{\min} = \inf \{ x : \Pr[X \geq x] > 0 \}$.
%
% For all $\epsilon > 0$ there exist $\psi < 1$ and $K > 0$ such that:
% \begin{itemize}
%  \item For all $q \in [Kn^{-3/4},1-\epsilon]$,
% 	\[\Ex [ \shapley_n(q) ]= \frac{1}{n} \Ex_{x \sim X^n_{\max}} [\frac{x}{\Ex [X_{\leq x}]}] \pm O(\psi^{n^{1/4}}). \]
%  \item For all $q \in [Kn^{-3/4},1-Kn^{-1/3}]$,
% 	\[ \Ex [\shapley_1(q) ]= \frac{1}{n} \Ex_{x \sim X^n_{\min}} [\frac{x}{\Ex [X_{\geq x}]}] \pm O(\psi^{n^{1/4}}). \]
% \end{itemize}
% %\end{theorem}
% \end{restatable}

In particular, we can determine the limiting values of $n\Ex[\shapley_n(Q)]$ and $n\Ex[\shapley_1(Q)]$.

\begin{corollary}
	\label{cor:uniform-limit}

Let $X$ be a random variable satisfying the requirements of Theorem~\ref{thm:iid}.
Suppose that $q \in (0,1)$ and (for the first statement) $\chi_{\max} < \infty$, where $\chi_{\max} = \sup \{ x : \Pr[X \leq x] > 0 \}$. Then for $Q = qn\Ex[X]$,
	\begin{align*}
		\lim_{n\to\infty} n\Ex [\shapley_n(Q)] &= \frac{\chi_{\max}}{\Ex [X]}, \\
		\lim_{n\to\infty} n\Ex [\shapley_1(Q)] &= \frac{\chi_{\min}}{\Ex [X]}.
	\end{align*}
If $\chi_{\max} = \infty$ then as $n\to\infty$ we have $n\Ex[\shapley_n(qn\Ex[X])] \to \infty$ and
\[ n\Ex[\shapley_n(Q)] \sim \Ex_{x\sim X^n_{\max}} [\frac{x}{\Ex [X_{\leq x}]}]. \]
\end{corollary}
\begin{proof}
Let $Q = qn\Ex[X]$.
We start by proving the result for $\shapley_1$. For large enough $n$, it is the case that $Q \in [n^{1/4}, (n-n^{2/3})\Ex[X]]$, and so we can apply the theorem to obtain
 \[ \Ex[\shapleyx_{\min}(Q)] = \frac{1}{n} \Ex_{x \sim X^n_{\min}} \big[\frac{x}{\Ex [X_{\geq x}]}\big] \pm O(\xi^{n^{1/4}}). \]
As $n\to\infty$, the random variable $X^n_{\min}$ tends to the constant $\chi_{\min}$, and so since $x/\Ex[X_{\geq x}]$ is continuous, this implies the formula for $\shapley_1$.

The proof for $\shapley_n$ is similar but has an extra complication. Suppose first that $\chi_{\max} < \infty$. Take $m$ large enough so that $\Ex[X_{\leq m}] > q \Ex[X]$. For large enough $n$, it is the case that $Q \in [n^{1/4}, (n-n^{2/3})\Ex[X_{\leq m}]]$, and so we can apply the theorem to obtain
 \[ \Ex[\shapleyx_{\max}(Q)] = \frac{1}{n} \Ex_{x \sim X^n_{\max}} \big[\frac{x}{\Ex [X_{\leq x}]}\big] \pm O(\xi^{n^{1/4}}). \]
As before, the random variable $X^n_{\max}$ tends to the constant distribution $\chi_{\max}$, and this implies the formula for $\shapley_n$.

Finally, suppose that $\chi_{\max} = \infty$. For the same choice of $m$ and for large enough $n$, the theorem implies that
\[
 n\Ex[\shapley_n(Q)] = \Ex_{x\sim X^n_{\max}} [\frac{x}{\Ex [X_{\leq x}]}] \pm o(1).
\]
In order to finish the proof, it remains to show that the expectation on the right tends to infinity. Indeed, this expectation is at least $\Ex[X^n_{\max}]/\Ex[X]$, and it is not hard to check that $\Ex[X^n_{\max}] \to \infty$.
\end{proof}

\section{The Uniform Distribution}
\label{sec:uniform}
Given the general result given in Theorem~\ref{thm:iid}, we now
demonstrate one application of it on the uniform distribution $U(a,b)$
, where $0 \leq a < b$. It is easy to see that $\chi_{\min} = a$ and
$\chi_{\max} = b$ in this case. When $X = U(a,b)$,
Corollary~\ref{cor:uniform-limit} shows that for $q \in (0,1)$,
$n\shapley_n(qn\Ex[X]) \to 2b/(a+b)$ and $n\shapley_1(qn\Ex[X]) \to 2a/(a+b)$.
% , and Corollary~\ref{cor:other-limit}, if true, would show that for
% $p,q \in (0,1)$, $n\shapley_{pn}(q) \to (2a+2(b-a)p)/(a+b)$.
Theorem~\ref{thm:iid} implies the following more precise statement,
illustrated by Figure~\ref{fig:uniform} for the case $X = U(0,1)$:
\begin{theorem}
  \label{thm:uniform-uniform-general}
  Let $X = U(a,b)$, where $0 \leq a < b$. For all $\epsilon > 0$ there
  exist $\psi < 1$ and $K > 0$ such that the following hold:
  \begin{itemize}
  \item For all $q \in [n^{-3/4},1-\epsilon]$,
    \[ \Ex [\shapley_n(qn\Ex[X])] = \frac{1}{b-a} \int_a^b
    \left(\frac{t-a}{b-a}\right)^{n-1} \frac{2t}{a+t} \, dt \pm
    O(\psi^{n^{1/4}}). \] We can evaluate the integral as a series:
    \[
    \frac{1}{b-a} \int_a^b \left(\frac{t-a}{b-a}\right)^{n-1}
    \frac{2t}{a+t} \, dt = \frac{2b}{a+b} \frac{1}{n} - \frac{2a}{a+b}
    \sum_{d=1}^\infty \left(\frac{b-a}{a+b}\right)^d
    \frac{d!}{n(n+1)\cdots(n+d)}.
    \]
    In particular, when $a = 0$ the integral is equal to
    $\frac{2}{n}$.
  \item For all $q \in [n^{-3/4},1-Kn^{-1/3}]$,
    \[ \Ex [\shapley_1(qn\Ex[X]) ]= \frac{1}{b-a} \int_a^b
    \left(\frac{b-t}{b-a}\right)^{n-1} \frac{2t}{b+t} \, dt \pm
    O(\psi^{n^{1/4}}). \] We can evaluate the integral as a series:
    \[
    \frac{1}{b-a} \int_a^b \left(\frac{b-t}{b-a}\right)^{n-1}
    \frac{2t}{b+t} \, dt = \frac{2a}{a+b} \frac{1}{n} - \frac{2b}{a+b}
    \sum_{d=1}^\infty \left(\frac{a-b}{a+b}\right)^d
    \frac{d!}{n(n+1)\cdots(n+d)}.
    \]
    In particular, when $a = 0$ the integral is equal to
    \[
    2\sum_{d=1}^\infty \frac{(-1)^{d+1} d!}{n(n+1)\cdots(n+d)} =
    \frac{2}{n(n+1)} - \frac{4}{n(n+1)(n+2)} + \cdots.
    \]
  \end{itemize}
\end{theorem}
\begin{proof}
  If $a = 0$ then $\Ex[X_{\leq x}] = x/2$ and so
  $x = O(\Ex[X_{\leq x}])$ for all $x \in [a,b]$.  If $a > 0$ then for
  $x \geq a$ we trivially have $\Ex[X_{\leq x}] \geq a$ and so
  $x = O(\Ex[X_{\leq x}])$ again for all $x \in [a,b]$.  In both cases
  we see that Theorem~\ref{thm:iid} applies.  We have
  $\chi_{\min} = a$, $\chi_{\max} = b$, $\Ex[X] = (a+b)/2$,
  $\Ex[X_{\leq x}] = (a+x)/2$ and $\Ex[X_{\geq x}] = (b+x)/2$.

  The distribution of $X_{\max}^n$ is given by
  \[ \Pr[X_{\max}^n \leq t] = \Pr[X \leq t]^n =
  \left(\frac{t-a}{b-a}\right)^n. \]
  The corresponding density function is the derivative
  $\frac{n}{b-a} \left(\frac{t-a}{b-a}\right)^{n-1}$.  The formula for
  $\shapley_n(qn\Ex[X])$ follows from
  \[
  \Ex_{x \sim X_{\max}^n} [\frac{x}{\Ex [X_{\leq x}]}] = \Ex_{x \sim
    X_{\max}^n} [\frac{2x}{a+x}] = \frac{1}{b-a} \int_a^b
  n\left(\frac{t-a}{b-a}\right)^{n-1} \frac{2t}{a+t} \, dt.
  \]

  Similarly, the distribution of $X_{\min}^n$ is given by
  \[ \Pr[X_{\min}^n \geq t] = \Pr[X \geq t]^n =
  \left(\frac{b-t}{b-a}\right)^n. \]
  The corresponding density function is the negated derivative
  $\frac{n}{b-a} \left(\frac{b-t}{b-a}\right)^{n-1}$.  The formula for
  $\shapley_1(qn\Ex[X])$ follows from
  \[
  \Ex_{x \sim X_{\min}^n} [\frac{x}{\Ex [X_{\geq x}]}] = \Ex_{x \sim
    X_{\min}^n} [\frac{2x}{b+x}] = \frac{1}{b-a} \int_a^b
  n\left(\frac{b-t}{b-a}\right)^{n-1} \frac{2t}{b+t} \, dt.
  \]

  % Before we estimate the integrals, notice that when $a = 0$ we have
  % $\Ex [X_{\leq x}] = x/2$, and so
  % $\Ex_{x \sim X_{\max}^n} [\frac{x}{\Ex [X_{\leq x}]}] = 2$. This
  % explains the exact formula for the first integral in the case
  % $a = 0$.

  % We proceed to evaluate the integrals, starting with the first
  % one. The basic observation is
  % \[
  % \frac{1}{b-a} \int_a^b \left(\frac{t-a}{b-a}\right)^{n-1} \, dt =
  % \left. \frac{1}{n} \left(\frac{t-a}{b-a}\right)^n \right|_1^b =
  % \frac{1}{n}.
  % \]
  % This implies that
  % \begin{align*}
  %   \frac{1}{b-a} \int_a^b \left(\frac{t-a}{b-a}\right)^{n-1} \frac{b-t}{b-a} \, dt &=
  %                                                                                     \frac{1}{b-a} \int_a^b \left(\frac{t-a}{b-a}\right)^{n-1} \, dt -  \frac{1}{b-a} \int_a^b \left(\frac{t-a}{b-a}\right)^n \, dt \\ &= \frac{1}{n} - \frac{1}{n+1} = \frac{1}{n(n+1)}.
  % \end{align*}
  % Similarly we can prove a more general formula, valid for all
  % integers $d \geq 0$:
  % \begin{equation} \label{eq:uniform:1} \frac{1}{b-a} \int_a^b
  %   \left(\frac{t-a}{b-a}\right)^{n-1} \left(\frac{b-t}{b-a}\right)^d
  %   \, dt = \frac{d!}{n(n+1)\cdots(n+d)}.
  % \end{equation}
  We proceed to evaluate the integrals, starting with the first one.
  The basic observation is
\begin{equation} \label{eq:uniform:1}
  \frac{1}{b-a} \int_a^b \left(\frac{t-a}{b-a}\right)^{n-1}
  \left(\frac{b-t}{b-a}\right)^d \, dt =
  \int_0^1 s^{n-1} (1-s)^d \, dt =
  \frac{d!}{n(n+1)\cdots(n+d)},
\end{equation}
  using the substitution $s=(t-a)/(b-a)$ and the Beta integral.
  A simple calculation shows that
  \[
  \frac{t}{a+t} = \frac{b}{a+b} - \frac{a}{a+b} \sum_{d=1}^\infty
  \left(\frac{b-t}{a+b}\right)^d.
  \]
  Therefore, using~\eqref{eq:uniform:1},
  \begin{align*}
    \frac{1}{b-a} \int_a^b \left(\frac{t-a}{b-a}\right)^{n-1} \frac{2t}{a+t} \, dt &=
                                                                                     \frac{2}{b-a} \int_a^b \left(\frac{t-a}{b-a}\right)^{n-1} \left[\frac{b}{a+b} - \frac{a}{a+b} \sum_{d=1}^\infty \left(\frac{b-t}{a+b}\right)^d \right] \, dt \\ &=
                                                                                                                                                                                                                                                       \frac{2b}{a+b} \frac{1}{n} - \frac{2a}{a+b} \sum_{d=1}^\infty \left(\frac{b-a}{a+b}\right)^d \frac{d!}{n(n+1)\cdots(n+d)}.
  \end{align*}

  The second integral can be evaluated in the same way. Alternatively,
  substitute $(a,b)=(b,a)$ in the formula for the first integral to
  obtain
  \begin{align*}
    \frac{1}{b-a} \int_a^b \left(\frac{b-t}{b-a}\right)^{n-1} \frac{2t}{b+t} \, dt &=
                                                                                     \frac{1}{a-b} \int_b^a \left(\frac{t-b}{a-b}\right)^{n-1} \frac{2t}{b+t} \, dt \\ &=
                                                                                                                                                                         \frac{2a}{a+b} \frac{1}{n} - \frac{2b}{a+b} \sum_{d=1}^\infty \left(\frac{a-b}{a+b}\right)^d \frac{d!}{n(n+1)\cdots(n+d)}. %\qedhere
  \end{align*}
\end{proof}

% Figure~\ref{fig:uniform} illustrates
% Theorem~\ref{thm:uniform-uniform-general} for the case $U(0,1)$.

%%% Local Variables:
%%% mode: latex
%%% TeX-master: "wvg-iid"
%%% End:

\section{The Exponential Distribution}
\label{sec:exponential}

In this section we analyze the exponential distribution $X = \Exp(1)$, whose bounds are $\chi_{\min} = 0$ and $\chi_{\max} = \infty$. Corollary~\ref{cor:uniform-limit} shows that $n\shapley_1(qn\Ex[X]) \to 0$ for all $q \in (0,1)$ and implies that $n\shapley_n(qn\Ex[X]) \to \infty$.
%Since $\Pr[X \leq t] = 1-e^{-t}$, if Conjecture~\ref{thm:iid-other} holds then Corollary~\ref{cor:other-limit} shows that for $p,q \in (0,1)$, $n\shapley_{pn}(q) \to -\log(1-p)$.
Theorem~\ref{thm:iid} implies the following more precise statement, illustrated by Figure~\ref{fig:exponential}:
\begin{theorem}
	\label{thm:exponential}
	Let $X = \Exp(1)$. For all $\epsilon > 0$ there exist $\psi < 1$ and $K > 0$ such that the following hold:
\begin{itemize}
        \item For all $q \in [n^{-3/4},1-\epsilon]$,
	\[ \Ex [\shapley_n(qn\Ex[X])] = \int_0^\infty (1-e^{-x})^n \frac{x}{e^x - (1+x)} \, dx + O(\psi^{n^{1/4}}), \]
        and the integral satisfies
\[
 \int_0^\infty (1-e^{-x})^n \frac{x}{e^x - (1+x)} \, dx = \frac{\log n + \gamma}{n} + O\left(\frac{\log^2 n}{n^2}\right).
\]
        \item For all $q \in [n^{-3/4},1-Kn^{-1/3}]$,
	\[ \Ex [\shapley_1(qn\Ex[X]) ]= \int_0^\infty e^{-nx} \frac{x}{x+1} \, dx + O(\psi^{n^{1/4}}), \]
        and the integral satisfies
\[ \int_0^\infty e^{-nx} \frac{x}{x+1} \, dx = \frac{1}{n^2} - O\left(\frac{1}{n^3}\right). \]
\end{itemize}
\end{theorem}
\begin{proof}
 Notice first that
\[
 \Ex[X_{\leq x}] = \frac{\int_0^x e^{-t} t \, dt}{\int_0^x e^{-t} \, dt} =\frac{1-(x+1)e^{-x}}{1-e^{-x}}.
\]
 For small $x$ the numerator is $1-(x+1)(1-x+x^2/2+O(x^3)) = (3/2) x^2 + O(x^3)$ and the denominator is $x + O(x^2)$, and so $\Ex[X_{\leq x}] \sim (3/2) x$. We conclude that $x = O(\Ex[X_{\leq x}])$ for $x$ near $0$, and so Theorem~\ref{thm:iid} applies.

 Using the formula for $\Ex[X_{\leq x}]$, we have
\[
\phi(x) := \frac{x}{\Ex[X_{\leq x}]} = \frac{x(1-e^{-x})}{1-(x+1)e^{-x}}.
\]
It is easy to calculate $\Pr[X^n_{\max} \leq x] = (1-e^{-x})^n$, and so the density of $X^n_{\max}$ is $n(1-e^{-x})^{n-1} e^{-x}$. We conclude that
\begin{align*}
\Ex_{x \sim X^n_{\max}} [\phi(x)] &= n\int_0^\infty (1-e^{-x})^n \frac{x}{e^x - (1+x)} \, dx \\ &=
n\int_0^1 (1-t)^n \log \frac{1}{t} \frac{dt}{1-(t+t\log\tfrac{1}{t})}.
\end{align*}
In order to estimate the integral, write
\[
 \int_0^1 (1-t)^n \log \frac{1}{t} \frac{dt}{1-(t+t\log\tfrac{1}{t})} = I_n + J_n + K_n,
\]
where $I_n,J_n,K_n$ are given by
\begin{align*}
 I_n &= \int_0^1 (1-t)^n \log \frac{1}{t} \, dt, \\
 J_n &= \int_0^1 (1-t)^n \log \frac{1}{t} (t+t\log\tfrac{1}{t}) \, dt, \\
 K_n &= \int_0^1 (1-t)^n \log \frac{1}{t} (t+t\log\tfrac{1}{t})^2 \frac{dt}{1-(t+t\log\tfrac{1}{t})}.
\end{align*}
Surprisingly, we can calculate $I_n,J_n$ exactly in terms of the harmonic numbers $H_n$:
\begin{align}
 I_n &= \frac{H_{n+1}}{n+1} = \frac{1}{n+1} \sum_{i=0}^n \frac{1}{i+1}, \label{eq:In} \\
 J_n &= \frac{1}{(n+1)(n+2)} \sum_{i=0}^n \frac{2}{i+2} H_i - \frac{i^2-i-4}{(i+1)(i+2)^2}. \label{eq:Jn}
\end{align}
In order to get the formula for $I_n$, notice first that $t+t\log\tfrac{1}{t}$ is an antiderivative of $\log\tfrac{1}{t}$. This immediately implies that $I_0=1$, and for $n>0$, integration by parts givs
\begin{align*}
 I_n &= \int_0^1 (1-t)^n \log \frac{1}{t} \, dt \\ &=
 \left. (1-t)^n \left(t+t\log\frac{1}{t}\right) \right|_0^1 + n\int_0^1(1-t)^{n-1}t\left(1+\log\frac{1}{t}\right) \, dt \\ &=
 n\int_0^1 [(1-t)^{n-1}-(1-t)^n]\log\frac{1}{t} \, dt + n\int_0^1 (1-t)^{n-1} t \, dt \\ &=
 n(I_{n-1} - I_n) + \frac{1}{n+1},
\end{align*}
using the formula for the Beta integral. This shows that $(n+1)I_n = nI_{n-1} + 1/(n+1)$, which implies formula~\eqref{eq:In}. Formula~\eqref{eq:Jn} is proved along the same lines.
It is well-known that $H_n = \log n + \gamma + O(1/n)$, and this shows that
\[
 I_n = \frac{H_{n+1}}{n+1} = \frac{\log (n+1) + \gamma + O(1/n)}{n+1} = \frac{\log n + \gamma}{n} + O\left(\frac{\log n}{n^2}\right).
\]
Similarly, using the integral $\int_1^n \frac{2\log m}{m} \, dm =
\log^2 n$ to estimate the corresponding series, we obtain
\[
 J_n = \frac{1}{\Theta(n^2)} \sum_{i=0}^n O\left(\frac{\log n}{n}\right) = O\left(\frac{\log^2 n}{n^2}\right).
\]
It remains to estimate $K_n$. We break $K_n$ into two parts, $L_n = \int_0^{1/e}$ and $M_n = \int_{1/e}^1$, which we bound separately. Since $t+t\log\frac{1}{t}$ is increasing, when $t \leq 1/e$ we have $t+t\log\frac{1}{t} \leq 2/e < 1$, and so
\[
 L_n = \int_0^{1/e} (1-t)^n \log \frac{1}{t} (t+t\log\tfrac{1}{t})^2 \frac{dt}{1-(t+t\log\tfrac{1}{t})} \leq \frac{1}{1-2/e} J_n = O\left(\frac{\log^2 n}{n^2}\right).
\]
When $t \geq 1/e$, we have $\log \frac{1}{t} \leq 1$ and so
\[
 M_n \leq 4\int_{1/e}^1 (1-t)^n \frac{dt}{1-(t+t\log\tfrac{1}{t})} =
 4\int_0^{1-1/e} s^n \frac{ds}{s+(1-s)\log(1-s)},
\]
where we applied the substitution $s = 1-t$. Taylor expansion shows that $s+(1-s)\log(1-s) \geq s^2/2$, and so
\[
 M_n \leq 8\int_0^{1-1/e} s^{n-2} \, ds = 8\frac{(1-1/e)^{n-1}}{n-1}.
\]
We conclude that $J_n + K_n = O(\log^2 n/n^2)$, and so
\[
 \int_0^1 (1-t)^n \log \frac{1}{t} \frac{dt}{1-(t+t\log\tfrac{1}{t})} = \frac{\log n + \gamma}{n} + O\left(\frac{\log^2 n}{n^2}\right).
\]

% In order to estimate this integral, note first that
% \[
% \int_0^\infty (1-e^{-x})^n \frac{x}{e^x} \, dx = \frac{H_{n+1}}{n+1},
% \]
% where $H_r$ is the $r$th harmonic number. This surprising formula is easily proved by induction. In order to estimate the error, note first that
% \[ \frac{1}{e^x - (1+x)} - \frac{1}{e^x} = \frac{1+x}{e^x(e^x-(1+x))}. \]
% Hence the difference between the actual integral and $H_{n+1}/(n+1)$ is
% \begin{align*}
% &\hphantom{=} \int_0^\infty (1-e^{-x})^n \frac{x(1+x)}{e^x(e^x - (1+x))} \, dx \\ &\leq
% \int_0^1 (1-e^{-x})^n \frac{x(1+x)}{e^x(e^x-(1+x))} \, dx +
% \int_1^\infty (1-e^{-x})^n \frac{4}{e^x} \, dx.
% \end{align*}
% The second integral can be calculated to be exactly $4/(n+1)$. We can estimate the first integral by
% \[
% \int_0^1 (1-e^{-x})^n \frac{x(1+x)}{e^x(e^x-(1+x))} \, dx \leq
% \int_0^1 x^n \frac{2x}{x^2/2} \, dx = \frac{4}{n}.
% \]
% In total, the error term is at most $8/n$. The stated asymptotic estimate follows from the well-known estimate $H_{n+1} = \log (n+1) + O(1)$.

We move on to calculate $\Ex [\shapley_1(qn\Ex[X])]$. We have
\[
\Ex[X_{\geq x}] = \frac{\int_x^{\infty} e^{-t} t \, dt}{\int_x^{\infty} e^{-t} \, dt} =\frac{(x+1)e^{-x}}{e^{-x}} = x+1.
\]
It is easy to calculate $\Pr[X^n_{\min} \geq x] = e^{-nx}$, and so the density of $X^n_{\min}$ is $ne^{-nx}$. We conclude that
\[
\Ex_{x \sim X^n_{\min}} \big[\frac{x}{\Ex[X_{\geq x}]}\big] = n\int_0^\infty e^{-nx} \frac{x}{x+1} \, dx.
\]
This gives us the stated formula. In order to estimate the integral, note that
\begin{align*}
\int_0^\infty e^{-nx} \frac{x}{x+1} \, dx &= \int_0^\infty e^{-nx} \, dx - \int_0^\infty e^{-nx} \frac{dx}{x+1} \\ &=
\frac{1}{n} - e^n \int_1^\infty  e^{-nx} \frac{dx}{x} =
\frac{1}{n} - e^n \int_n^\infty \frac{e^{-x}}{x} \, dx
\end{align*}
The latter integral is an exponential integral, and its asymptotic expansion is
\[ \int_n^\infty \frac{e^{-x}}{x} \, dx = e^{-n} \left(\frac{1}{n} - \frac{1}{n^2} + O\left(\frac{1}{n^3}\right)\right). \]
We conclude that
\[
 \int_0^\infty e^{-nx} \frac{x}{x+1} \, dx = \frac{1}{n^2} - O\left(\frac{1}{n^3}\right). \qedhere
\]
\end{proof}

\begin{figure}%[H]
	\centering
	\includegraphics[width=.8\textwidth]{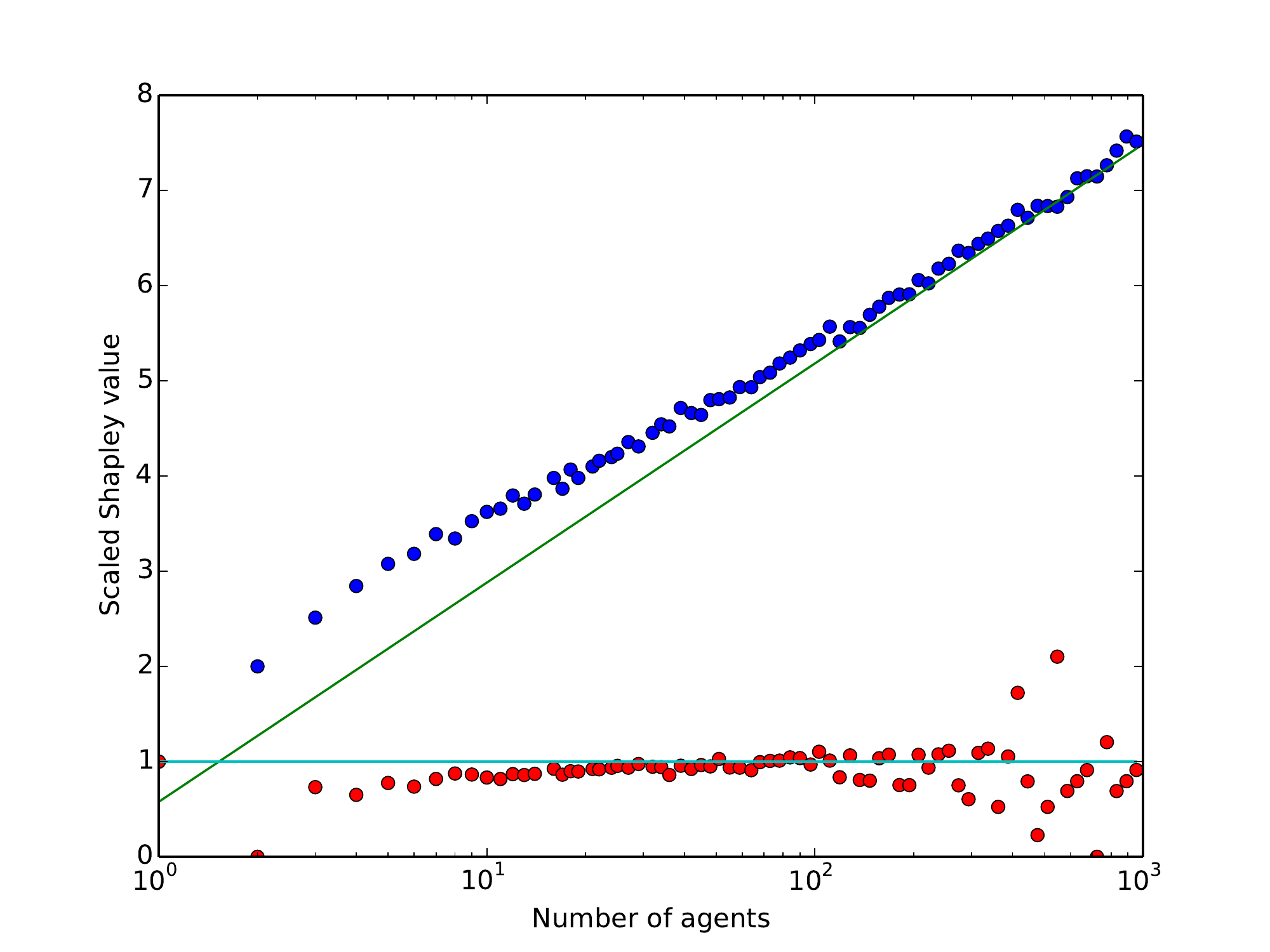} \caption{Shapley values for the normalized iid model with $X = \Exp(1)$ and various $n$ of both maximal and minimal agent, multiplied by $n$ and by $n^2$, respectively, at the quota $q = 1/2$. Results of $10^6$ experiments. The experimental results are compared to the predictions of Theorem~\ref{thm:exponential}: the maximal Shapley value is compared against $(\log n + \gamma)/n$, and the minimal Shapley value is compared against $1/n^2$.} \label{fig:exponential}
\end{figure}

\section{Proving Theorem \ref{thm:iid}}
\label{sec:main-proof}
Recall the statement of Theorem~\ref{thm:iid}:

\thmiid*

%\begin{proof}

Recall that we generated the weights $w_1,\ldots,w_n$ according to the following process. First, we generate the sequence $x_1,\ldots,x_n$ by generating $n$ samples from the distribution $X$.
The sequence $w_1,\ldots,w_n$ consists of the values $x_1,\ldots,x_n$ sorted in increasing order.
%We then compute $S = x_1+\cdots+x_n$, and the sequence $w_1,\ldots,w_n$ consists of the values $x_1/S,\ldots,x_n/S$ sorted in increasing order.
It will be simpler to analyze the original sequence $x_1,\ldots,x_n$ and some values derived from it:
\begin{itemize}
	\item $x_{\max} = \max(x_1,\ldots,x_n)$, the corresponding distribution is $X^n_{\max}$, and the corresponding Shapley value (with respect to $x_1,\ldots,x_n$) is $\shapleyx_{\max}$.
	\item $x_{\min} = \min(x_1,\ldots,x_n)$, the corresponding distribution is $X^n_{\min}$, and the corresponding Shapley value (with respect to $x_1,\ldots,x_n$) is $\shapleyx_{\min}$.
\end{itemize}
Recall also that we defined
\begin{align*}
 \chi_{\min} &= \inf \{ x : \Pr[X \geq x] > 0 \}, & \chi_{\max} &= \sup \{ x : \Pr[X \leq x] > 0 \}.
\end{align*}

The crux of the proof is the following formula for the Shapley values of the original sequence $x_1,\ldots,x_n$.

\begin{lemma}
	\label{lem:uniform-formula} For any quota value $Q$,
	\begin{align}
		\Ex[\shapleyx_{\max}(Q)] &= \Ex_{x \sim X^n_{\max}} \left[\frac{1}{n} \sum_{i=1}^n \Pr_{y_1,\ldots,y_{n-1} \sim X_{\leq x}} \left[ \sum_{j=1}^{i-1} y_j \in [Q-x,Q) \right] \right], \label{eq:xmax} \\
		\Ex[\shapleyx_{\min}(Q)] &= \Ex_{x \sim X^n_{\min}} \left[\frac{1}{n} \sum_{i=1}^n \Pr_{y_1,\ldots,y_{n-1} \sim X_{\geq x}} \left[ \sum_{j=1}^{i-1} y_j \in [Q-x,Q) \right] \right]. \label{eq:xmin}
	\end{align}
\end{lemma}
\begin{proof}
	The proofs of both formulas are similar, so we only prove the first one. We show that conditioned on $x_{\max} = x$,
	\[ \Ex[\shapleyx_{\max}(Q)] = \frac{1}{n} \Pr_{y_1,\ldots,y_{n-1} \sim X_{\leq x}} \left[ \sum_{j=1}^{i-1} y_j \in [Q-x,Q) \right]. \]
	We can assume without loss of generality that $x_{\max} = x_n$. Given only this data, the variables $x_1,\ldots,x_{n-1}$ are distributed independently according to $X_{\leq x_n}$. Therefore

	%Since a permutation of $\{1,\ldots,n\}$ can be obtained by taking a permutation of $\{1,\ldots,n-1\}$ and putting $n$ in a random place,
	\begin{align*}
		\Ex[\shapleyx_n(Q)] &= \Ex_{\pi \in S_n} [x_n \text{ is pivotal in } x_{\pi_1},\ldots,x_{\pi_n}] \\
		&= \frac{1}{n} \sum_{i=1}^n \Ex_{\substack{\pi \in S_n\colon \\
		\pi_i = n}} \Pr[x_n \text{ is pivotal in } x_{\pi_1},\ldots,x_{\pi_n}] \\
		&= \frac{1}{n} \sum_{i=1}^n \Ex_{\substack{\pi \in S_n\colon \\
		\pi_i = n}} \Pr\left[\sum_{j=1}^{i-1} x_{\pi_j} \in [Q-x,Q)\right].
	\end{align*}
	Here \emph{pivotal} is always with respect to the threshold $Q$. Since $x_1,\ldots,x_{n-1}$ are independent and identically distributed, $x_{\pi_1},\ldots,x_{\pi_{i-1}}$ are distributed identically to $y_1,\ldots,y_{i-1}$, proving~\eqref{eq:xmax}. Formula~\eqref{eq:xmin} is proved along similar lines.
\end{proof}

\subsection{Estimating the formulas} \label{sec:uniform-estimate}
Recall that our main approach is to use an analogy to renewal processes, in which each of the agent weights can be thought of as a renewal `step' and that furthermore, estimating the expected number of points that land within the interval $[q-w_n,q)$ will be used for proving the formulas for the highest Shapley value (and similarly for the lowest Shapley value).

The first step towards achieving this goal is to extend the sums in Lemma~\ref{lem:uniform-formula} to infinite sums. Estimating these infinite sums will be done using the following lemma, which is relevant to renewal processes with exponentially decaying renewal time distributions.
\begin{proposition}
\label{pro:renewal}
Suppose $Y$ is a continuous distribution supported on $[0,\infty)$ whose density function $f(t)$ is bounded by $Ce^{-\lambda t}$ for some $C,\lambda > 0$.
Furthermore, suppose that $Y$ is $\intervals$-piecewise differentiable-monotone: the support of $Y$ can be partitioned into $\intervals$ many intervals, on each of which $f$ is differentiable and monotone (non-increasing or non-decreasing).
There exist constants $B,\gamma > 0$, depending only on $C,\lambda,\intervals$, such that for $Q \geq 0$,
\[ \sum_{i=1}^\infty \Pr_{y_1,\ldots,y_{i-1} \sim Y} \left[ \sum_{j=1}^{i-1} y_j < Q \right] = \frac{Q}{\Ex Y} + \frac{\Ex(Y^2)}{2(\Ex Y)^2} + \epsilon, \text{ where }|\epsilon| \leq Be^{-\gamma Q}. \]
\end{proposition}
We conjecture that Proposition~\ref{pro:renewal} can be strengthened by removing the condition that $f$ is piecewise differentiable-monotone. This condition is only used in one place in the proof, Lemma~\ref{lem:riemann-lebesgue}.

This form of the renewal theorem differs from others in the literature in that the error term is uniform over \emph{all} distributions with given decay.
We prove Proposition~\ref{pro:renewal} in Section~\ref{sec:sound}.

In order to utilize this proposition for the estimation the sums in Lemma~\ref{lem:uniform-formula}, we need to restrict the value of $x$ in $X_{\geq x}$ and $X_{\leq x}$. For $m \in (\chi_{\min},\chi_{\max})$, we say that $Y$ is an \emph{$m$-reasonable} random variable if either $Y = X_{\geq x}$ for $x \leq m$ or $Y = X_{\leq x}$ for $x \geq m$. These variables enjoy the following properties.

\begin{lemma} \label{lem:reasonable-props}
Let $m \in (\chi_{\min},\chi_{\max})$, and define $\mu = \min(\Pr[X \leq m], \Pr[X \geq m]) > 0$. Then the density $g$ of every $m$-reasonable random variable $Y$ satisfies $g(t) \leq \frac{C}{\mu} e^{-\lambda t}$, and $\Ex[Y] \leq \frac{C}{\mu\lambda} \cdot \frac{1}{\lambda}$, $\Ex[Y^2] \leq \frac{C}{\mu\lambda} \cdot \frac{2}{\lambda^2}$.
%\begin{align*}
%\Ex[Y] &\leq \frac{C}{\mu\lambda} \cdot \frac{1}{\lambda}, & \Ex[Y^2] &\leq \frac{C}{\mu\lambda} \cdot \frac{2}{\lambda^2}.
%\end{align*}
Also, $\Pr[Y \geq t] \leq \frac{C}{\lambda\mu} e^{-\lambda t}$.
\end{lemma}
\begin{proof}
When $Y = X_{\geq x}$, we have $g(t) = 0$ for $t < x$ and $g(t) = f(t)/\Pr[X \geq x] \leq f(t)/\mu \leq (C/\mu)e^{-\lambda t}$. A similar calculation shows that $g(t) \leq (C/\mu)e^{-\lambda t}$ when $Y = X_{\leq x}$. The bound on the density, in turn, implies the bound on the moments, using the formulas $\Ex[\Exp(\lambda)] = 1/\lambda$, $\Ex[\Exp(\lambda)^2] = 2/\lambda^2$:
\[ \Ex[Y] = \int_0^\infty g(t) t \, dt \leq \frac{C}{\lambda \mu} \int_0^\infty \lambda e^{-\lambda t} t \, dt = \frac{C}{\lambda \mu} \cdot \frac{1}{\lambda}. \]
The bound on $\Ex[Y^2]$ is obtained similarly. The bound on $\Pr[Y \geq t]$ is obtained using a similar calculation applied to $\Pr[\Exp(\lambda) \geq t] = e^{-\lambda t}$.
\end{proof}

We can now apply Proposition~\ref{pro:renewal}.
\begin{corollary}
	\label{cor:renewal} Let $Y$ be an $m$-reasonable random variable, for some $m \in (\chi_{\min},\chi_{\max})$. Then for some $\gamma < 1$ depending only on $m$ and for all $Q \geq 0$,
	\[ \sum_{i=1}^\infty \Pr_{y_1,\ldots,y_{i-1} \sim Y} \left[ \sum_{j=1}^{i-1} y_j < Q \right] = \frac{Q}{\Ex Y} + \frac{\Ex(Y^2)}{2(\Ex Y)^2} \pm O(\gamma^Q). \]
	In particular, for all $x \geq 0$ and $Q \geq x$,
	\[ \sum_{i=1}^\infty \Pr_{y_1,\ldots,y_{i-1} \sim Y} \left[ \sum_{j=1}^{i-1} y_j \in [Q-x,Q) \right] = \frac{x}{\Ex Y} \pm O(\gamma^{Q-x}). \]
\end{corollary}
\begin{proof}
The first statement of the corollary follows directly from Proposition~\ref{pro:renewal}, given Lemma~\ref{lem:reasonable-props}; note that any $m$-reasonable variable is $\intervals$-piecewise differentiable-monotone. The second statement follows by applying the first statement to $Q$ and to $Q-x$, and subtracting the two estimates.
\end{proof}

The corollary affords us with a good estimate of the sums in Lemma~\ref{lem:uniform-formula}, when extended from $n$ to $\infty$. In order to estimate the actual sums, we estimate the tail from $n+1$ to $\infty$.
\begin{lemma}
	\label{lem:uniform-tail} Let $Y$ be an $m$-reasonable random variable, for some $m \in (\chi_{\min},\chi_{\max})$. For some $\delta < 1$ depending only on $m$ and for all $Q \leq (n-n^{2/3}) \Ex Y$,
	\[ \sum_{i=n+1}^\infty \Pr_{y_1,\ldots,y_{i-1} \sim Y} \left[ \sum_{j=1}^{i-1} y_j < Q \right] = O(\delta^{n^{1/4}}), \]
	where the constant in $O(\cdot)$ depends only on $m$.
\end{lemma}
\begin{proof}
The idea of the proof is to apply Bernstein's inequality to show that it is highly improbable that the sum of $n$ variables (or more) distributed according to $Y$ be significantly smaller than $n\Ex[Y]$. One complication is that Bernstein's inequality only applies to bounded random variables, whereas $Y$ could be unbounded. In order to fix this, we choose a cut-off $M$ and consider the random variable $Z = \min(Y,M)$ instead. Note that
\begin{gather*}
\Ex[Y] - \Ex[Z] = \int_M^\infty g(t) (t-M) \, dt \leq \frac{C}{\lambda\mu} \int_M^\infty \lambda e^{-\lambda t} (t-M) \, dt \\
= \frac{C}{\lambda\mu} \lambda e^{-\lambda M} \int_0^\infty \lambda e^{-\lambda t} t \, dt = \frac{C}{\lambda^2\mu} e^{-\lambda M} = O(e^{-\lambda M}).
\end{gather*}
Therefore, if we choose $M = K\log n$ for some constant $K$ depending only on $X$ and $m$ then we can ensure that $\Ex[Y] - \Ex[Z] \leq (1/(2n^{1/3})) \Ex[X_{\leq m}]$ and so $\Ex[Y] - \Ex[Z] \leq (1/(2n^{1/3})) \Ex[Y]$.

For the rest of this proof, let $y_1,y_2,\ldots \sim Y$ be independent copies of $Y$, and let $z_1,z_2,\ldots \sim Z$ be independent copies of $Z$.
Bernstein's inequality implies that when $Q \leq (n+k) \Ex[Z]$ (which always holds, as we show below),
\begin{align*}
\Pr\left[ \sum_{j=1}^{n+k} y_j < Q \right] &\leq \Pr\left[ \sum_{j=1}^{n+k} z_j < Q \right] \\ &\leq
\exp -\frac{((n+k) \Ex[Z] - Q)^2/2}{(n+k)\var[Z] + ((n+k) \Ex[Z] - Q) (M/3)} \\ &\leq
\exp -\frac{((n+k) \Ex[Z] - Q)^2/2}{(n+k)(\var[Z] + (M/3)\Ex[Z])} \\ &\leq
\exp -\frac{((n+k) \Ex[Z] - Q)^2}{(n+k)(3M^2)}.
\end{align*}
Note that
\[
n\Ex[Z] - Q \geq \left(n-\frac{n^{2/3}}{2}\right)\Ex[Y] - Q \geq \frac{1}{2} n^{2/3} \Ex[Y],
\]
and therefore, using $\Ex[Z] \geq \Ex[Y]/2$,
\[
((n+k) \Ex[Z] - Q)^2 \geq (k\Ex[Z] + \tfrac{1}{2} n^{2/3} \Ex[Y])^2 \geq \tfrac{1}{4} \Ex[Y]^2 (k+n^{2/3})^2.
\]
When $k \leq n$, we have
\[
\exp -\frac{((n+k) \Ex[Z] - Q)^2}{(n+k)(3M^2)} \leq
\exp -\frac{\tfrac{1}{4} \Ex[Y]^2 n^{4/3}}{(2n)(3M^2)} =
\exp -\frac{\Ex[Y]^2 n^{1/3}}{24M^2}.
\]
When $k \geq n$, we have
\[
\exp -\frac{((n+k) \Ex[Z] - Q)^2}{(n+k)(3M^2)} \leq
\exp -\frac{\tfrac{1}{4} \Ex[Y]^2 k^2}{(2k)(3M^2)} =
\exp -\frac{\Ex[Y]^2 k}{24M^2}.
\]
Therefore
\begin{align*}
\sum_{k=0}^\infty \Pr\left[ \sum_{j=1}^{n+k} y_j < Q \right] &\leq
n \exp -\frac{\Ex[Y]^2 n^{1/3}}{24M^2} + \sum_{k=n}^\infty \exp -\frac{\Ex[Y]^2 k}{24M^2} \\ &\leq
n \exp -\frac{\Ex[Y]^2 n^{1/3}}{24M^2} + \frac{\exp -\frac{\Ex[Y]^2 n}{24M^2}}{1-\exp -\frac{\Ex[Y]^2}{24M^2}}.
\end{align*}
Since $\Ex[Y] \geq \Ex[X_{\leq m}]$ and $M = O(\log n)$, we can conclude that
\[
\sum_{k=0}^\infty \Pr\left[ \sum_{j=1}^{n+k} y_j < Q \right] \leq
n e^{-\Omega\left(\frac{n^{1/3}}{\log^2 n}\right)} + O\left(\log^2 n e^{-\Omega\left(\frac{n}{\log^2 n}\right)}\right),
\]
implying the lemma.
\end{proof}

Combining this with Corollary~\ref{cor:renewal}, we obtain the following estimate.
\begin{corollary}
	\label{cor:uniform-estimate} Let $Y$ be an $m$-reasonable random variable, for some $m \in (\chi_{\min},\chi_{\max})$. Then for some $\zeta < 1$ depending only on $m$, for all $x \geq 0$ and for all $Q \in [x,(n-n^{2/3})\Ex [Y]]$,
	\[ \sum_{i=1}^n \Pr_{y_1,\ldots,y_{n-1} \sim Y} \left[ \sum_{j=1}^{i-1} y_j \in [Q-x,Q) \right] = \frac{x}{\Ex [Y]} \pm O(\zeta^{n^{1/4}} + \zeta^{Q-x}). \]
\end{corollary}
\begin{proof}
	Clearly
	\[ \sum_{i=n+1}^\infty \Pr_{y_1,\ldots,y_{i-1} \sim Y} \left[ \sum_{j=1}^{i-1} y_j \in [Q-x,Q) \right] \leq \sum_{i=n+1}^\infty \Pr_{y_1,\ldots,y_{i-1} \sim Y} \left[ \sum_{j=1}^{i-1} y_j < Q \right] = O(\delta^{n^{1/4}}), \]
	using Lemma~\ref{lem:uniform-tail}. Therefore Corollary~\ref{cor:renewal} implies that
	\[ \sum_{i=1}^n \Pr_{y_1,\ldots,y_{n-1} \sim Y} \left[ \sum_{j=1}^{i-1} y_j \in [Q-x,Q) \right] = \frac{x}{\Ex [Y]} \pm O(\delta^{n^{1/4}} + \gamma^{Q-x}). \]
	The corollary follows by taking $\zeta = \max(\delta,\gamma)$.
\end{proof}

Using this estimate, we can estimate the sums in Lemma~\ref{lem:uniform-formula}. The idea is to focus on the case in which the variable $X_{\leq x}$ or $X_{\geq x}$ is $m$-reasonable.

\begin{lemma} \label{lem:uniform-estimate}
 	Let $m \in (\chi_{\min},\chi_{\max})$. For some $\xi < 1$ depending on $m$ and for all $Q \in [n^{1/4},(n-n^{2/3})\Ex [X_{\leq m}]]$,
	\[ \Ex[\shapleyx_{\max}(Q)] = \frac{1}{n} \Ex_{x \sim (X^n_{\max})_{\geq m}} \big[\frac{x}{\Ex [X_{\leq x}]}\big] \pm O(\xi^{n^{1/4}}). \]
	Similarly, for all $Q \in [n^{1/4}, (n-n^{2/3})\Ex [X]]$,
	\[ \Ex[\shapleyx_{\min}(Q)] = \frac{1}{n} \Ex_{x \sim X^n_{\min}} \big[\frac{x}{\Ex [X_{\geq x}]}\big] \pm O(\xi^{n^{1/4}}). \]
\end{lemma}
\begin{proof}
We start with the first formula. Let $\mu_{\leq} = \Pr[X \leq m]$ and $\mu_{\geq} = \Pr[X \geq m] = 1 - \mu_{\leq}$. Clearly $\Pr[X^n_{\max} \leq m] = \mu_{\leq}^n$. Using Lemma~\ref{lem:uniform-formula}, we get for any $M > m$ that
\[ \Ex[\shapleyx_{\max}(Q)] = \Ex_{x \sim (X^n_{\max})_{\geq m, \leq M}}\Big[ \frac{1}{n} \sum_{i=1}^n \Pr_{y_1,\ldots,y_{n-1} \sim X_{\leq x}} \left[ \sum_{j=1}^{i-1} y_j \in [Q-x,Q) \right] \Big] \pm O(\mu_{\leq}^n+\Pr[X^n_{\max}>M]]). \]
Corollary~\ref{cor:uniform-estimate} implies that for all $Q \in [M, (n-n^{2/3})\Ex [X_{\leq m}]]$,
\[ \Ex[\shapleyx_{\max}(Q)] = \frac{1}{n} \Ex_{x \sim (X^n_{\max})_{\geq m, \leq M}}\big[ \frac{x}{\Ex [X_{\leq x}]}\big] \pm O(\zeta^{n^{1/4}} + \zeta^{Q-M} + \mu_{\leq}^n + \Pr[X^n_{\max}>M]). \]
(We need the restriction $x \leq M$ so that we can bound the term $\zeta^{Q-M}$.)
We proceed to estimate the main term, aiming to remove the restriction $x \leq M$. Define $\phi(x) = x/\Ex[X_{\leq x}]$, and note that $\phi(x) \leq x/\Ex[X_{\leq m}]$ when $x \geq m$. We have
\begin{align*}
\Ex_{x \sim (X^n_{\max})_{\geq m}}[\phi(x)] &=
\Pr[(X^n_{\max})_{\geq m} > M] \Ex_{x \sim (X^n_{\max})_{> M}}[\phi(x)] \\ &+
(1-\Pr[(X^n_{\max})_{\geq m} > M]) \Ex_{x \sim (X^n_{\max})_{\geq m,\leq M}}[\phi(x)].
\end{align*}
Note that $\Pr[(X^n_{\max})_{\geq m} > M]$ is at most the probability that the maximum of $n$ variables distributed $X_{\geq m}$ is more than $M$. Lemma~\ref{lem:reasonable-props} and a union bound show that this probability is at most $n \frac{C}{\lambda\mu\Pr[X\geq m]} e^{-\lambda M} = O(ne^{-\lambda M})$. Together with $\phi(x) \leq M/\Ex[X_{\leq m}]$ whenever $m \leq x \leq M$, this shows that
\[
\Ex_{x \sim (X^n_{\max})_{\geq m}}[\phi(x)] = \Ex_{x \sim (X^n_{\max})_{\geq m,\leq M}}[\phi(x)] +
\Pr[(X^n_{\max})_{\geq m} > M] \Ex_{x \sim (X^n_{\max})_{> M}}[\phi(x)] - O(nMe^{-\lambda M}).
\]
Suppose that for some integer $K>0$, $M \geq K\lambda^{-1} \log n$. Lemma~\ref{lem:reasonable-props} shows that $\Pr[X > t] \leq \frac{C}{\lambda\mu} e^{-\lambda t}$, and so
\begin{align*}
&\hphantom{=} \Pr[(X^n_{\max})_{\geq m} > M] \Ex_{x \sim (X^n_{\max})_{>M}}[\phi(x)] \\ &\leq
\frac{1}{\Ex[X_{\leq m}]} \Pr[(X^n_{\max})_{\geq m} > M] \Ex_{x \sim (X^n_{\max})_{>M}}[x] \\ &\leq
\frac{1}{\Ex[X_{\leq m}]} \Pr[(X^n_{\max})_{\geq m} > K\lambda^{-1} \log n] \Ex_{x \sim (X^n_{\max})_{>K\lambda^{-1} \log n}}[x] \\ &\leq
\frac{1}{\Ex[X_{\leq m}]} \sum_{r=K}^\infty \lambda^{-1} (r+1) \log n \Pr[\lambda^{-1} r\log n \leq X^n_{\max} \leq \lambda^{-1} (r+1)\log n] \\ &\leq
\frac{C}{\lambda \mu \Ex[X_{\leq m}]} \sum_{r=K}^\infty \lambda^{-1} (r+1) \log n e^{-r\log n} \\ &=
\frac{O(C)}{\lambda \mu \Ex[X_{\leq m}]} \lambda^{-1} \log n \frac{K}{n^K}.
\end{align*}
If $M \geq 2\lambda^{-1} \log n$ then we can choose $K = \lfloor \lambda M/\log n \rfloor \geq \lambda M/(2\log n)$ and so $n^K \geq e^{(\lambda/2)M}$ and $K/n^K \leq \lambda M e^{-(\lambda/2)M}$. Therefore
\[
 \Ex_{x \sim (X^n_{\max})_{\geq m}}[\phi(x)] = \Ex_{x \sim (X^n_{\max})_{\geq m,\leq M}}[\phi(x)] \pm O(nMe^{-(\lambda/2)M}).
\]
Lemma~\ref{lem:reasonable-props} and a union bound show that $\Pr[X^n_{\max}>M] = O(ne^{-\lambda M})$, and we deduce that for $M \geq 2\lambda^{-1} \log n$,
\begin{align*}
\Ex[\shapleyx_{\max}(Q)] &= \frac{1}{n} \Ex_{x \sim (X^n_{\max})_{\geq m, \leq M}}\big[ \frac{x}{\Ex [X_{\leq x}]}\big] \pm O(\zeta^{n^{1/4}} + \zeta^{Q-M} + \mu_{\leq}^n + ne^{-\lambda M}) \\ &=
\frac{1}{n} \Ex_{x \sim (X^n_{\max})_{\geq m}}\big[ \frac{x}{\Ex [X_{\leq x}]}\big] \pm O(\zeta^{n^{1/4}} + \zeta^{Q-M} + \mu_{\leq}^n + Me^{-(\lambda/2) M}).
\end{align*}
Choosing $M = Q/2$, we deduce that
\[ \Ex[\shapleyx_{\max}(Q)] = \frac{1}{n} \Ex_{x \sim (X^n_{\max})_{\geq m}}\big[ \frac{x}{\Ex [X_{\leq x}]}\big] \pm O(\zeta^{n^{1/4}} + \zeta^{Q/2} + \mu_{\leq}^n + \log n e^{-(\lambda/4)Q}). \]
This implies the formula in the statement of the lemma, with $\xi > \max(\sqrt{\zeta}, e^{-\lambda/4}, \mu_{\leq}, \mu_{\geq})$ (we need $\mu_{\geq}$ for the other part of the lemma).

We continue with the second formula. As before, we have
\[ \Ex[\shapleyx_{\min}(Q)] = \Ex_{x \sim (X^n_{\min})_{\leq m}}\Big[ \frac{1}{n} \sum_{i=1}^n \Pr_{y_1,\ldots,y_{n-1} \sim X_{\geq x}} \left[ \sum_{j=1}^{i-1} y_j \in [Q-x,Q) \right] \Big]\pm O(\mu_{\geq}^n). \]
Corollary~\ref{cor:uniform-estimate} implies that for all $Q \in [m, (n-n^{2/3})\Ex [X]]$,
\begin{equation} \label{eq:shap_min}
		\shapleyx_{\min}(Q) = \frac{1}{n} \Ex_{x \sim (X^n_{\min})_{\leq m}}\big[ \frac{x}{\Ex [X_{\geq x}]}\big] \pm O(\zeta^{n^{1/4}} + \zeta^Q + \mu_{\geq}^n).
\end{equation}
Now, as by definition $\Pr[X^n_{\min}\geq m]=\mu^n_{\geq}$, we have that
\begin{align}
  \Ex_{x \sim X^n_{\min}} \big[\frac{x}{\Ex [ X_{\geq x}]} \big] &= \mu^n_{\geq} \cdot \Ex_{x \sim (X^n_{\min})_{\geq m}} \big[\frac{x}{\Ex [X_{\geq x}]}\big] + (1-\mu^n_{\geq}) \cdot \Ex_{x \sim (X^n_{\min})_{\leq m}}\big[ \frac{x}{\Ex [X_{\geq x}]}\big] \nonumber \\ &=
  \Ex_{x \sim (X^n_{\min})_{\leq m}} \frac{x}{\Ex [X_{\geq x}]} + \mu^n_{\geq} \cdot \big( \Ex_{x \sim (X^n_{\min})_{\geq m}} \big[\frac{x}{\Ex [X_{\geq x}]}\big] - \Ex_{x \sim (X^n_{\min})_{\leq m}}\big[ \frac{x}{\Ex [X_{\geq x}]}\big] \big) \nonumber \\ &=
  \Ex_{x \sim (X^n_{\min})_{\leq m}} \frac{x}{\Ex [X_{\geq x}]} + O(\mu^n_{\geq} ),
\end{align}
where the last line follows from the fact that $\frac{x}{\Ex X_{\geq x}} \leq 1$.
Combining this with~\eqref{eq:shap_min} gives the second formula in the statement of the lemma.
\end{proof}

We can now prove out main result.

\thmiid*
\begin{proof}
The formula for $\Ex[\shapleyx_{\min}(Q)]$ is already stated in Lemma~\ref{lem:uniform-estimate}, so we only prove the formula for $\Ex[\shapleyx_{\max}(Q)]$.
Given $\epsilon > 0$, we choose $m$ large enough so that $\Ex[X_{\leq m}] > (1-\epsilon) \Ex[X]$. For large enough $n$, the condition $Q \leq (1-\epsilon) n\Ex[X]$ implies the condition $Q \leq (n-n^{2/3})\Ex[X]$.

Since $x = O(\Ex[X_{\leq x}])$ for $x$ near $\chi_{\min}$, the expectation $\Ex_{x \sim (X^n_{\max})_{\leq m}} [\frac{x}{\Ex [ X_{\leq x}]}]$ converges.
Therefore we similarly have
\begin{align}
  \Ex_{x \sim X^n_{\max}} \big[\frac{x}{\Ex [ X_{\leq x}]} \big] &= \mu^n_{\leq} \cdot \Ex_{x \sim (X^n_{\max})_{\leq m}} \big[\frac{x}{\Ex [X_{\leq x}]}\big] + (1-\mu^n_{\leq}) \cdot \Ex_{x \sim (X^n_{\max})_{\geq m}}\big[ \frac{x}{\Ex [X_{\leq x}]}\big] \nonumber \\ &=
  \Ex_{x \sim (X^n_{\max})_{\geq m}} \frac{x}{\Ex [X_{\leq x}]} + \mu^n_{\geq} \cdot \big( \Ex_{x \sim (X^n_{\max})_{\leq m}} \big[\frac{x}{\Ex [X_{\leq x}]}\big] - \Ex_{x \sim (X^n_{\max})_{\geq m}}\big[ \frac{x}{\Ex [X_{\leq x}]}\big] \big) \nonumber \\ &=
  \Ex_{x \sim (X^n_{\max})_{\geq m}} \frac{x}{\Ex [X_{\leq x}]} + O(\mu^n_{\leq} ),
\end{align}
implying the formula for $\Ex[\shapleyx_{\max}(Q)]$.
\end{proof}

\section{Conjuctural extensions} \label{sec:other}

\subsection{Normalized iid model} \label{sec:normalized-iid-model}

Theorem~\ref{thm:iid} predicts the values of the minimal and maximal Shapley values in the natural iid model.
We conjecture that a similar theorem holds for the normalized iid model.

\begin{conjecture}
	\label{thm:iid-normalized}
Let $X$ be a non-negative continuous random variable whose density function $f$ satisfies $f(t) \leq Ce^{-\lambda t}$ for some $C,\lambda > 0$, and additionally, the support of $f$ can be partitioned into finitely many intervals on which $f$ is differentiable and monotone.
Furthermore, for the first statement, assume also that $x = O(\Ex[X_{\leq x}])$ for $x$ near $\chi_{\min}$, where $\chi_{\min} = \inf \{ x : \Pr[X \geq x] > 0 \}$.

Generate weights $w_1,\ldots,w_n$ according to the normalized iid model: generate $n$ i.i.d.\ samples $x_1,\ldots,x_n$ of $X$, let $S = x_1+\cdots+x_n$, and let $w_1,\ldots,w_n$ consist of the values $x_1/S,\ldots,x_n/S$ sorted in increasing order.

For all $\epsilon > 0$ there exist $\psi < 1$ and $K > 0$ such that:
\begin{itemize}
 \item For all $q \in [Kn^{-3/4},1-\epsilon]$,
	\[\Ex [ \shapley_n(q) ]= \frac{1}{n} \Ex_{x \sim X^n_{\max}} [\frac{x}{\Ex [X_{\leq x}]}] \pm O(\psi^{n^{1/4}}). \]
 \item For all $q \in [Kn^{-3/4},1-Kn^{-1/3}]$,
	\[ \Ex [\shapley_1(q) ]= \frac{1}{n} \Ex_{x \sim X^n_{\min}} [\frac{x}{\Ex [X_{\geq x}]}] \pm O(\psi^{n^{1/4}}). \]
\end{itemize}
\end{conjecture}

This implies the following corollary, whose proof is very similar to the proof of Corollary~\ref{cor:uniform-limit}.

\begin{conjecture}
	\label{cor:uniform-limit-normalized}
Let $X$ be a random variable satisfying the requirements of Theorem~\ref{thm:iid-normalized}.
Suppose that $q \in (0,1)$ and (for the first statement) $\chi_{\max} < \infty$, where $\chi_{\max} = \sup \{ x : \Pr[X \leq x] > 0 \}$. Then
	\begin{align*}
		\lim_{n\to\infty} n\Ex [\shapley_n(q)] &= \frac{\chi_{\max}}{\Ex [X]}, \\
		\lim_{n\to\infty} n\Ex [\shapley_1(q)] &= \frac{\chi_{\min}}{\Ex [X]}.
	\end{align*}
If $\chi_{\max} = \infty$ then as $n\to\infty$ we have $n\Ex[\shapley_n(q)] \to \infty$ and
\[ n\Ex[\shapley_n(q)] \sim \Ex_{x\sim X^n_{\max}} [\frac{x}{\Ex [X_{\leq x}]}]. \]
\end{conjecture}

These conjectures are supported by our experiments for the uniform and exponential distributions, appearing in Figure~\ref{fig:uniform} (uniform distribution) and Figure~\ref{fig:exponential} (exponential distribution). While these experiments were performed using the normalized iid model, their results match those of the natural iid model.

Theorem~\ref{thm:iid} gives a strong estimate for $\Ex[\shapley_n(Q)]$ and $\Ex[\shapley_1(Q)]$. In terms of the natural iid model, Conjecture~\ref{thm:iid-normalized} predicts a strong estimate for the quantities $\Ex[\shapley_n(q \sum_{i=1}^n x_i)]$ and $\Ex[\shapley_1(q \sum_{i=1}^n x_i)]$. Since $\sum_{i=1}^n x_i$ is strongly concentrated around $n\Ex[X]$, and $\Ex[\shapley_i(Q)]$ is concentrated around some limiting value $\Phi_i$ (for $i = 1,n$) for $Q \approx qn\Ex[X]$, we expect $\Ex[\shapley_i(q \sum_{i=1}^n x_i)] \approx \Phi_i$. In order to show this, it suffices to prove that for $S \approx n\Ex[X]$ and $Q \approx qn\Ex[X]$, we have $\Ex[\shapley_i(Q)|\sum_{i=1}^n x_i = S] \approx \Phi_i$. Unfortunately, at the moment we cannot prove this estimate.

\subsection{Other Shapley values} \label{sec:other-shapley-values}

Theorem~\ref{thm:iid} and Conjecture~\ref{thm:iid-normalized} describe the behavior of the Shapley values corresponding to the minimal and maximal agents. It is natural to ask how the Shapley values in between behave. Based on the proof of Theorem~\ref{thm:iid}, we can formulate a conjecture for the behavior of the non-extreme Shapley values. Since the formulation is cleaner in the normalized iid model, we present it in that model.
%We conjecturally extend Theorem~\ref{thm:iid} to handle other Shapley values.
For $p \in (0,1)$, let $\shapley_{pn}(q)$ be the Shapley value corresponding to the $pn$th order statistics, let $X^n_{pn}$ be the distribution of the $pn$th order statistics (in both cases rounding $pn$ arbitrarily to an integer), and let $X_{\mixture{p}{x}}$ be the random variable which is a mixture of $X_{\leq x}$ (with probability $p$) and $X_{\geq x}$ (with probability $1-p$).

\begin{conjecture}
	\label{thm:iid-other}
Let $X$ be a random variable satisfying the requirements of Theorem~\ref{thm:iid-normalized}.
For all $p \in (0,1)$ there exist $\psi < 1$ and $K > 0$ such that for all $q \in [Kn^{-3/4},1-Kn^{1/3}]$,
\[ \Ex[\shapley_{pn}(q)] = \frac{1}{n} \Ex_{x \sim X^n_{pn}} [\frac{x}{\Ex [X_{\mixture{p}{x}}]}] \pm o\left(\frac{1}{n}\right). \]
\end{conjecture}

This conjecture determines the limiting value of $n\Ex [\shapley_{pn}(q)]$:

\begin{corollary}
	\label{cor:other-limit}
Suppose that $p,q \in (0,1)$. Then
\[ \lim_{n \to \infty} n\Ex[\shapley_{pn}(q)] = \frac{x}{\Ex[X]}, \text{ where } \Pr[X \leq x] = p. \]
\end{corollary}
\begin{proof}
 As $n\to\infty$, the distribution of $X^n_{pn}$ tends to the constant $x$. At that point, we have
\begin{align*}
\Ex[X_{\mixture{p}{x}}] &= p\Ex[X_{\leq x}] + (1-p) \Ex[X_{\geq x}] \\ &= \Pr[X \leq x] \Ex[X_{\leq x}] + \Pr[X \geq x] \Ex[X_{\geq x}] \\ &= \Ex[X]. \qedhere
\end{align*}
\end{proof}

As an illustration, we apply Corollary~\ref{cor:other-limit} to the uniform distribution $U(a,b)$ and to the exponential distribution $\Exp(1)$. When $X = U(a,b)$, we have $\Ex[X] = \frac{a+b}{2}$ and $\Pr[X \leq x] = \frac{x-a}{b-a}$, and so $\Pr[X \leq x] = p$ for $x = (1-p)a + pb$. Therefore Corollary~\ref{cor:other-limit} implies that $n\Ex[\shapley_{pn}(q)] \to \frac{2(1-p)a + 2pb}{a+b}$ for all $q \in (0,1)$. In particular, when $X = U(0,1)$ the corollary implies that $n\Ex[\shapley_{pn}(q)] \to 2p$ for all $q \in (0,1)$.

When $X = \Exp(1)$, we have $\Ex[X] = 1$ and $\Pr[X \leq x] = 1-e^{-x}$, and so $\Pr[X \leq x] = p$ for $x = -\log(1-p)$. Therefore Corollary~\ref{cor:other-limit} implies that $n\Ex[\shapley_{pn}(q)] \to -\log(1-p)$ for all $q \in (0,1)$.

Corresponding experimental results shown in Figure~\ref{fig:other-limit} support Corollary~\ref{cor:other-limit} and so Conjecture~\ref{thm:iid-other}. %We explain how we arrived at the conjecture in Section~\ref{sec:conjecture-explanation}.

\begin{figure}%[H]
	\centering
	\includegraphics[width=.8\textwidth]{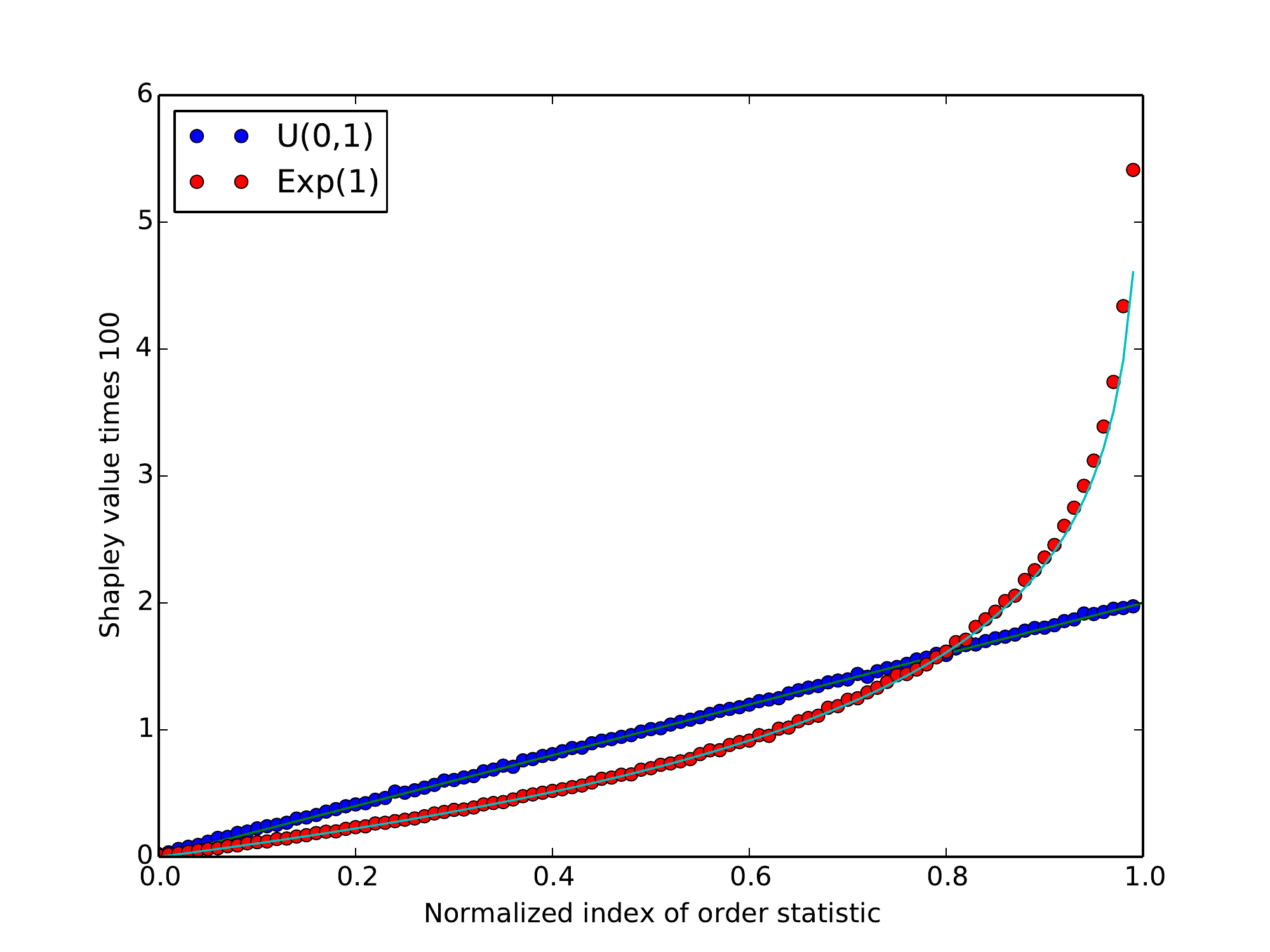} \caption{All Shapley values for $X=U(0,1)$ and $X=\Exp(1)$ (both in the normalized iid model) and the setting $n = 100$ at the quota $q = 1/2$, normalized by $n$. Results of $10^6$ experiments. The experimental results are compared to the predictions of Corollary~\ref{cor:other-limit}: $\shapley_{pn}(q) = 2p$ for $X=U(0,1)$ and $\shapley_{pn}(q) = -\log(1-p)$ for $X=\Exp(1)$.} \label{fig:other-limit}
\end{figure}

%\subsection{On Conjecture~\ref{thm:iid-other}} \label{sec:conjecture-explanation}

%Section~\ref{sec:main-other} gives a conjectured generalization of Theorem~\ref{thm:iid} to other order statistics. In this section we briefly explain how we arrived at the conjecture.

In the remainder of this section, we briefly described how Conjecture~\ref{thm:iid-other} follows by extending the ideas of Section~\ref{sec:main-proof}.

A straightforward generalization of Lemma~\ref{lem:uniform-formula} shows that
\begin{equation} \label{eq:other-binomial}
\Ex[\shapleyx_{p(n-1)}(Q)] = \frac{1}{n} \Ex_{x \sim X^n_{p(n-1)}} \left[\sum_{i=1}^n \Pr_{y_1,\ldots,y_{n-1}} \left[ \sum_{j=1}^{i-1} y_j \in [Q-x,Q) \right] \right],
\end{equation}
where $y_1,\ldots,y_{n-1}$ are chosen as follows: a random subset of $p(n-1)$ of them are chosen according to $X_{\leq x}$, and the rest are chosen according to $X_{\geq x}$. A generalization of Corollary~\ref{cor:uniform-estimate} shows that if we choose some $m_1,m_2$ satisfying $\chi_{\min} < m_1 < m_2 < \chi_{\max}$, then for all $x \in [m_1,m_2]$ we have
\begin{equation} \label{eq:other-prob}
\sum_{i=1}^n \Pr_{y_1,\ldots,y_{n-1} \sim X_{\mixture{p}{x}}} \left[ \sum_{j=1}^{i-1} y_j \in [Q-x,Q) \right] = \frac{x}{\Ex [X_{\mixture{p}{x}}]} \pm O(\zeta^{n^{1/4}} + \zeta^{Q-x}).
\end{equation}
The expressions in~\eqref{eq:other-binomial} and~\eqref{eq:other-prob} are very similar, the difference being that in~\eqref{eq:other-binomial} the number of positions distributed according to $X_{\leq x}$ is exactly $p(n-1)$, while in~\eqref{eq:other-prob} it is strongly concentrated around $p(n-1)$. It is therefore reasonable to conjecture that the analog of Lemma~\ref{lem:uniform-estimate} holds,
\[
\Ex[\shapleyx_{p(n-1)}(Q)] = \frac{1}{n} \Ex_{x \sim X^n_{p(n-1)}} [\frac{x}{\Ex[X_{\mixture{p}{x}}]}] + o\left(\frac{1}{n}\right).
\]
Given this, Conjecture~\ref{thm:iid-other} (in its version for the natural iid model) follows just as in the proof of Theorem~\ref{thm:iid}.

% The main difficulty in proving Conjecture~\ref{thm:iid-other} is converting between expectations under $y_1,\ldots,y_{n-1} \sim X_{\mixture{r}{x}}$ for $r \approx p$ (where $x \sim X^n_{pn}$) to expectations under $y_1,\ldots,y_{n-1}$ of which a random subset of $p(n-1)$ are distributed $X_{\leq x}$, and the rest are distributed $X_{\geq x}$. Such conversions occur in random graph theory (between the $G(n,p)$ and $G(n,M)$ models), leading us to believe that the conjecture is true, perhaps with a worse error term.

\section{Proving Proposition \ref{pro:renewal}} \label{sec:sound}
In this section, we complete the proof of Theorem~\ref{thm:iid} by proving Proposition~\ref{pro:renewal}.

The idea of the proof is to use the Mellin transform to write
\[ m(Q) \triangleq \sum_{i=1}^\infty \Pr_{y_1,\ldots,y_{i-1} \sim Y} \left[ \sum_{j=1}^{i-1} y_j < Q \right] = \frac{1}{2\pi i} \int_{c-i\infty}^{c+i\infty} \frac{e^{sQ}}{s(1 - \Ex[e^{-sY}])} \,ds, \]
where $c > 0$ is arbitrary.
The integrand has a double pole at $s = 0$ with residue $Q/\Ex Y + \Ex(Y^2)/2(\Ex Y)^2$, which gives rise to the main term in the proposition. The conditions on the distribution $Y$ imply that apart from the pole at $s = 0$, the integrand has no poles in some halfspace $\Re s \geq -\gamma$. Therefore we can move the line of integration to the left toward $-\gamma$, and then read off the error term. The exponential dependence comes from the numerator $e^{sQ}$.

In the rest of this section, we will assume that $Y$ is a continuous distribution supported on $[0,\infty)$ whose density function $f$ is bounded by $Ce^{-\lambda t}$ for some $C,\lambda > 0$. Whenever we use the term ``constant'', we mean a quantity depending only on the parameters $C,\lambda,\intervals$.

We start by proving the formula for $m(Q)$.
\begin{lemma} \label{lem:m-integral-formula}
For all $c > 0$,
\[ m(Q) = \frac{1}{2\pi i} \int_{c-i\infty}^{c+i\infty} \frac{e^{sQ}}{s(1 - \Ex[e^{-sY}])} \, ds. \]
\end{lemma}
\begin{proof}
It is well-known that
\[
\frac{1}{2\pi i} \int_{c-i \infty}^{c+i \infty} \frac{e^{sx}}{s} \, ds
= \begin{cases} 1 & x > 0, \\ 0 & x < 0. \end{cases}
\]
Therefore, letting $(y_k)_1^\infty \sim Y$,
\begin{align*}
m(Q) &= \sum_{k=1}^\infty \Pr_{y_1,\ldots,y_{k-1} \sim Y}
\left[ \sum_{j=1}^{k-1} y_j < Q \right] \\ &=
\sum_{k=1}^\infty \Ex[1_{Q - y_1+\cdots+y_{k-1} > 0}] \\ &=
\sum_{k=1}^\infty \Ex_{y_1,\ldots,y_{k-1} \sim Y} \frac{1}{2\pi i}
\int_{c-i\infty}^{c+i\infty} \frac{e^{s(Q-y_1-\cdots-y_{k-1})}}{s} \, ds \\ &=
\sum_{k=1}^\infty \frac{1}{2\pi i} \int_{c-i\infty}^{c+i\infty}
\frac{e^{sQ}}{s} \Ex[e^{-sy_1-\cdots-sy_{k-1}}] \, ds \\ &=
\sum_{k=1}^\infty \frac{1}{2\pi i} \int_{c-i\infty}^{c+i\infty}
\frac{e^{sQ}}{s} \Ex[e^{-sY}]^{k-1} \, ds \\ &=
\frac{1}{2\pi i} \int_{c-i\infty}^{c+i\infty}
\frac{e^{sQ}}{s(1-\Ex[e^{-sY}])} \, ds. \qedhere
\end{align*}
\end{proof}

The following elementary bounds will prove useful.

\begin{lemma} \label{lem:elem-bounds}
For each integer $k \geq 0$,
\[ \Ex[Y^k] \leq \frac{C}{\lambda} \frac{k!}{\lambda^k}. \]
For each $x < y$,
\[ \Pr[x \leq Y \leq y] \leq \frac{C}{\lambda} (e^{-\lambda x}-e^{-\lambda y}). \]
\end{lemma}
\begin{proof}
 When $f(u) = \lambda e^{-\lambda u}$, $Y$ is an exponential random variable with moments $\Ex[Y^k] = k!/\lambda^k$, and so in general
\[ \Ex[Y^k] \leq \frac{C}{\lambda} \int_0^\infty \lambda e^{-\lambda u} u^k \, du = \frac{C}{\lambda} \frac{k!}{\lambda^k}. \]
 The second bound follows by calculating $C \int_x^y e^{-\lambda u} \, du$.
\end{proof}

We proceed to show that in some halfspace $\Re s \geq -\gamma$, the integrand has no poles other than the double pole at $s = 0$. In fact, we will show more: in this halfspace, excepting a fixed neighborhood of zero, $|1 - \Ex[e^{-sY}]| = \Omega(1)$.

\begin{lemma} \label{lem:re-estimate}
For every $I > 0$ there are constants $\eta,R > 0$ such that
\[ \Re \Ex[e^{-sY}] \leq 1 - \eta \]
whenever $\Re s \geq -R$ and $|\Im s| \geq I$.
\end{lemma}
\begin{proof}
Let $s = -\alpha + i\beta$, where $\alpha \leq R$ (a constant to be determined) and $|\beta| \geq I$.
We have
\begin{align*}
 \Re \Ex[e^{-sY}] &= \int_0^\infty f(u) e^{\alpha u} \cos (\beta u) \, du \\ &=
 \int_0^\infty f(u)(\cos(\beta u) (e^{\alpha u}-1) + \cos(\beta u)) \, du \\ &\leq
 \int_0^\infty f(u) (e^{\alpha u}-1) \, du + \int_0^\infty f(u) \cos(\beta u) \, du \\ &=
 \int_0^\infty f(u) e^{\alpha u} \, du - \int_0^\infty f(u) (1-\cos(\beta u)) \, du.
\end{align*}
Denote the two terms by $A,B$, so that $\Re \Ex[e^{-sY}] = A - B$. We will show that $B \geq 2\eta$ for some $\eta > 0$, and that for $R \leq \lambda/2$,
\[ A \leq 1 + \frac{2CR}{\lambda^2}. \]
Taking $R > 0$ small enough, we can ensure that $A - B \leq 1 - \eta$.

We start with the bound on $A$. If $\alpha \leq 0$ then
\[
 A = \int_0^\infty f(u) e^{\alpha u} \, du \leq \int_0^\infty f(u) \, du = 1.
\]
If $0 < \alpha \leq R$ for some $R \leq \lambda/2$ then
\begin{align*}
 A = \Ex[e^{\alpha Y}] &= 1 + \sum_{k=1}^\infty \frac{\alpha^k}{k!} \Ex[Y^k] \\ &\leq
 1 + \frac{C}{\lambda} \sum_{k=1}^\infty \left(\frac{\alpha}{\lambda}\right)^k \\ &\leq
 1 + \frac{C}{\lambda} \sum_{k=1}^\infty \left(\frac{R}{\lambda}\right)^k \\ &=
 1 + \frac{C}{\lambda} \cdot \frac{R}{\lambda} \cdot 2,
\end{align*}
using Lemma~\ref{lem:elem-bounds}, proving the bound on $A$.

The bound on $B$ is slightly more complicated. For small $\epsilon > 0$, we have $1 - \cos (\beta u) < \epsilon$ only if $\beta u \in 2\pi\mathbb{Z} + [-K\sqrt{\epsilon},K\sqrt{\epsilon}]$, where $K$ is some universal constant. Therefore
\[
 B = \int_0^\infty f(u) (1-\cos(\beta u)) \, du \geq \epsilon \Pr[\beta Y \notin 2\pi\mathbb{Z} + [-K\sqrt{\epsilon},K\sqrt{\epsilon}]].
\]
Suppose without loss of generality that $\beta > 0$. For $\epsilon < 1/K^2$,
\begin{align*}
 &\hphantom{=} \Pr[\beta Y \in 2\pi\mathbb{Z} + [-K\sqrt{\epsilon},K\sqrt{\epsilon}]] \\ &=
 \Pr[Y \leq K\sqrt{\epsilon}/\beta] + \sum_{t=1}^\infty \Pr[|Y - (2 \pi/\beta) t| \leq K\sqrt{\epsilon}/\beta] \\ &\leq
 \frac{C}{\lambda} (1-e^{-(K\lambda/\beta)\sqrt{\epsilon}}) + \frac{C}{\lambda} (e^{(K\lambda/\beta)\sqrt{\epsilon}}-e^{-(K\lambda/\beta)\sqrt{\epsilon}}) \sum_{t=1}^\infty e^{-(2\pi \lambda/\beta)t} \\ &=
 \frac{C}{\lambda} (1-e^{-(K\lambda/\beta)\sqrt{\epsilon}}) + \frac{C}{\lambda} (e^{(K\lambda/\beta)\sqrt{\epsilon}}-e^{-(K\lambda/\beta)\sqrt{\epsilon}}) \frac{e^{-(2\pi \lambda/\beta)}}{1-e^{-(2\pi \lambda/\beta)}},
\end{align*}
using Lemma~\ref{lem:elem-bounds}.
The bound is monotone decreasing in $\beta$, and so since $\beta \geq I$, we can bound the probability by $O(\sqrt{\epsilon})$, where the hidden constant depends on $C,\lambda,I$. This shows that $B \geq \epsilon(1-O(\sqrt{\epsilon}))$. In particular, we can choose a small $\epsilon > 0$ such that $B \geq \epsilon/2$, proving the bound on $B$.
\end{proof}

The constant $I$ arises from the following lemma (we later choose $I = S/\sqrt{2}$).

\begin{lemma} \label{lem:estimate-around-zero}
There exist constants $S,\delta > 0$ such that for all $|s| \leq S$,
\[ |\Ex[e^{-sY}] - 1| \geq \delta |s|. \]
\end{lemma}
\begin{proof}
For any $\tau \geq 0$, using $f(u) \leq C$ we get
\[ \Ex[Y] \geq \tau \Pr[Y \geq \tau] = \tau (1 - \Pr[Y \leq \tau]) \geq \tau (1 - C\tau). \]
In particular, choosing $\tau = 1/(2C)$, we deduce $\Ex[Y] \geq 1/(4C)$.

Consider now any complex $s$. We have
\[
 \Ex[e^{-sY}] = 1 - \Ex[Y] s + \sum_{k=2}^\infty (-1)^k \frac{\Ex[Y^k]}{k!} s^k.
\]
Using Lemma~\ref{lem:elem-bounds}, we deduce that for $|s| \leq \lambda/2$,
\begin{align*}
 |\Ex[e^{-sY}] - 1| &\geq \Ex[Y] |s| - \sum_{k=2}^\infty \left(\frac{|s|}{\lambda}\right)^k \\ &\geq
 \frac{|s|}{4C} - \frac{2|s|^2}{\lambda^2} \\ &=
 |s| \left(\frac{1}{4C} - \frac{2|s|}{\lambda^2}\right).
\end{align*}
 If we choose $S = \min(\lambda/2,\lambda^2/(16C))$ then we deduce that $|\Ex[e^{-sY}]-1| \geq |s|/(8C)$.
\end{proof}

Combining the two lemmas, we obtain the following information on
$\Ex[e^{-sY}]$.

\begin{lemma} \label{lem:lb-char-func}
There are constants $\epsilon,\gamma,I > 0$ such that $\Ex[e^{-sY}] \neq 1$
whenever $\Re s \geq -\gamma$ and $s \neq 0$, and furthermore
\[ |\Ex[e^{-sY}] - 1| \geq \epsilon \min(|s|,1) \]
whenever $|\Re s| \leq \gamma$.
\end{lemma}
\begin{proof}
Let $S,\delta$ be the constants in Lemma~\ref{lem:estimate-around-zero}, let $I = S/\sqrt{2}$, and let $R,\eta$ be the constants in Lemma~\ref{lem:re-estimate}.
Define $\gamma = \min(R,I)$.
Suppose that $\Re s \geq -\gamma$ and $s \neq 0$. If $|\Im s| \geq I$ then $\Re \Ex[e^{-sY}] \leq 1 - \eta$, and in particular $\Ex[e^{-sY}] \neq 1$. If $|\Im s| \leq I$ then $|s| \leq \sqrt{\gamma^2 + I^2} \leq S$, and so $|\Ex [e^{-sY}]-1| \geq \delta |s|$, again showing that $\Ex[e^{-sY}] \neq 1$. This proves the first claim.

For the second claim, suppose that $|\Re s| \leq \gamma$. If $|\Im s| \geq I$ then $|\Ex[e^{-sY}] - 1| \geq |\Re \Ex[e^{-sY}] - 1| \geq \eta$. If $|\Im s| \leq I$ then as before $|s| \leq S$ and so $|\Ex[e^{-sY}]-1| \geq \delta |s|$. Taking $\epsilon = \min(\eta,\delta)$ proves the second claim.
\end{proof}

Next, we move the line of integration to the left in order to separate the main term $Q/\Ex Y + \Ex(Y^2)/2(\Ex Y)^2$ from the error term.

\begin{lemma} \label{lem:moving-left}
For all $Q > 0$ and all $0 < \beta < \gamma$, where $\gamma > 0$ is the constant from Lemma~\ref{lem:lb-char-func},
\[ m(Q) = \frac{Q}{\Ex Y} + \frac{\Ex(Y^2)}{(\Ex Y)^2} +\frac{1}{2\pi i}
\int_{-\beta-i\infty}^{-\beta+i\infty} \frac{e^{sQ}\Ex[e^{-sY}]}{s(1 - \Ex[e^{-sY}])} \,
ds. \]
\end{lemma}
\begin{proof}
Our starting point is the formula of Lemma~\ref{lem:m-integral-formula}, for $c = \beta$.
Lemma~\ref{lem:lb-char-func} shows that the only pole of the integrand in the strip $|s| \leq \beta$ is at $s = 0$.
Standard arguments (using the bound $|\Ex[e^{-sY}]-1| \geq \epsilon$ given by Lemma~\ref{lem:lb-char-func}) show that
\begin{align*}
m(Q) &= \frac{1}{2\pi i} \int_{-\beta-i\infty}^{-\beta+i\infty} \frac{e^{sQ}}{s(1 - \Ex[e^{-sY}])} \, ds +
\Res\left[\frac{e^{sQ}}{s(1 - \Ex[e^{-sY}])},s=0\right] \\ &=
\frac{1}{2\pi i} \int_{-\beta-i\infty}^{-\beta+i\infty} \frac{e^{sQ} \Ex[e^{-sY}]}{s(1 - \Ex[e^{-sY}])} \, ds +
\Res\left[\frac{e^{sQ}}{s(1 - \Ex[e^{-sY}])},s=0\right];
\end{align*}
the two integrals differ by the quantity
\[
\int_{-\beta-i\infty}^{-\beta+i\infty} \frac{e^{sQ}}{s} \, ds = 0.
\]
In order to compute the residue, write
\begin{align*}
\frac{e^{sQ}}{s(1 - \Ex[e^{-sY}])} &= \frac{1 + sQ + O(s^2)}{s^2(\Ex Y - \tfrac{1}{2} \Ex(Y^2) s + O(s^2))} \\ &=
\frac{(1 + sQ + O(s^2))\left(1 + \frac{\Ex(Y^2)}{2\Ex Y} s + O(s^2) \right)}{s^2 \Ex Y}.
\end{align*}
Calculating the coefficient of $s^{-1}$ in this expression completes the proof.
\end{proof}

In order to estimate the error term, we need to understand the
behavior of $\Ex[e^{-sY}]$ as $|s| \to \infty$.

\begin{lemma} \label{lem:riemann-lebesgue}
Suppose $\alpha = -\Re s$ satisfies $0 < \alpha \leq \lambda/2$. Then for some constant $K_\alpha$ depending on $\alpha$,
\[ |\Ex[e^{-sY}]| \leq \frac{K_\alpha}{\sqrt{|s|}}. \]
\end{lemma}

Note that this lemma is the only place in the proof in which we use the condition of piecewise differentiability-monotonicity.

\begin{proof}
Let $M = \frac{\log (|s|/\alpha)}{4/\lambda} \geq 0$. We will bound separately the two terms
\begin{align*}
 L &= \int_0^M e^{-su} f(u) \, du, &
 U &= \int_M^\infty e^{-su} f(u) \, du.
\end{align*}

We start with $L$. Integration by parts gives
\[
 \int_0^M e^{-su} f(u) \, du =
 \left. -\frac{f(u) e^{-su}}{s}\right|_0^M + \int_0^M f'(u) \frac{e^{-su}}{s} \, du.
\]
Therefore
\[
 |L| \leq C \frac{1+e^{(\lambda/2)M}}{|s|} + \frac{e^{(\lambda/2)M}}{|s|} \int_0^M |f'(u)| \, du.
\]
We can divide $[0,M]$ into $N \leq \intervals$ many intervals of monotonicity, say $I_1=(a_1,b_1),\ldots,I_N=(a_n,b_n)$, and bound
\[
 \int_0^M |f'(u)| \, du = \sum_{i=1}^N |f(b_i) - f(a_i)| \leq NC.
\]
We conclude that
\[
 |L| \leq C \frac{1+(\intervals+1)e^{(\lambda/2)M}}{|s|} \leq 3C\intervals\frac{e^{(\lambda/2)M}}{|s|} \leq \frac{3C\intervals}{\sqrt{\alpha|s|}}.
\]

We proceed to bound $U$:
\begin{gather*}
 |U| \leq C \int_M^\infty e^{(\alpha-\lambda)u} \, du \leq C \int_M^\infty e^{-(\lambda/2)u} \, du \\ =
 \frac{C}{\lambda/2} e^{-(\lambda/2)M} = \frac{C\sqrt{\alpha}}{(\lambda/2)\sqrt{|s|}}.
\end{gather*}

Altogether, we deduce
\[ |\Ex[e^{-sY}]| \leq |L| + |U| \leq \left(\frac{3C\intervals}{\sqrt{\alpha}} + \frac{2C\sqrt{\alpha}}{\lambda}\right) \frac{1}{\sqrt{|s|}}. \qedhere \]
\end{proof}

Finally, we estimate the error term.
\begin{lemma} \label{lem:error-term}
Let $\beta = \min(\gamma,\lambda/2) > 0$, where $\gamma > 0$ is the constant from Lemma~\ref{lem:lb-char-func}. We
have
\[ \left| \int_{-\beta-i\infty}^{-\beta+i\infty} \frac{e^{sQ}\Ex[e^{-sY}]}{s(1 - \Ex[e^{-sY}])} \, ds \right| = O(e^{-\beta Q}). \]
\end{lemma}
\begin{proof}
The triangle inequality shows that
\[\left| \int_{-\beta-i\infty}^{-\beta+i\infty} \frac{e^{sQ}\Ex[e^{-sY}]}{s(1 - \Ex[e^{-sY}])} \, ds \right| \leq
e^{-\gamma Q} \int_{-\beta-i\infty}^{-\beta+i\infty} \left|\frac{\Ex[e^{-sY}]}{s(1 - \Ex[e^{-sY}])}\right| \, ds. \]
Using Lemma~\ref{lem:riemann-lebesgue} and Lemma~\ref{lem:lb-char-func}, we can estimate the integrand:
\[
\left|\frac{\Ex[e^{-sY}]}{s(1 - \Ex[e^{-sY}])}\right| \leq K_\beta (\epsilon\min(1,\beta))^{-1} \frac{1}{|s|^{3/2}}.
\]
This estimate shows that the integral converges:
\[
 \int_{-\infty}^\infty \frac{1}{\sqrt{\beta^2+y^2}^{3/2}} \, dy \leq
 2\int_0^\infty \frac{\sqrt{2}^{3/2}}{(\beta+y)^{3/2}} \, dy = \frac{4\sqrt{2}^{3/2}}{\sqrt{\beta}}.
\]
The lemma follows.
\end{proof}

Proposition~\ref{pro:renewal} (with $\gamma:=\beta$) follows by combining Lemma~\ref{lem:moving-left} and Lemma~\ref{lem:error-term}.
\section{Conclusions} \label{sec:conclusions}
%\jo{Write a proper conclusions section here.}

In this paper we tackled the problem of estimating the smallest and largest Shapley values $\shapley_1,\shapley_n$ in the case in which agent weights are drawn i.i.d.\ according to a continuous distribution decaying exponentially. We gave a simple formula for the limiting values $n\shapley_1,n\shapley_n$ for all relevant quota values, and gave a more accurate formula for $\shapley_1,\shapley_n$ with an exponentially decaying error term.

Our results leave open several natural questions, besides resolving Conjecture~\ref{thm:iid-normalized} and Conjecture~\ref{thm:iid-other}, and strengthening Proposition~\ref{pro:renewal} (see page~\pageref{pro:renewal} for the latter). First, our results do not cover values of the quota which are very close to $0$ or to $n\Ex[X]$ (where $X$ is the distribution used to sample agent weights). Can we say anything about the behavior of the Shapley values in this regime? Second, our results only concern the expected Shapley values. Can we say anything about the actual distribution of the Shapley values, say by bounding the variance? Third, our results concern Shapley values. Can we extend the analysis to Banzhaf values?
\bibliographystyle{spbasic}
\bibliography{bib}

\end{document}